\documentclass[onecolumn, conference,letterpaper]{IEEEtran}

\usepackage[utf8]{inputenc} 
\usepackage[T1]{fontenc}
\usepackage{url}
\usepackage{ifthen}
\usepackage{cite}
\usepackage[cmex10]{amsmath} % Use the [cmex10] option to ensure complicance
\usepackage{booktabs}       % professional-quality tables
\usepackage{amsfonts}       % blackboard math symbols
\usepackage{nicefrac}       % compact symbols for 1/2, etc.
\usepackage{microtype}      % microtypography
\usepackage{xcolor}         % colors

\usepackage{enumitem}
\usepackage{bm}
\usepackage{amsmath}
\usepackage{amsthm}
\usepackage{caption}
\usepackage{amssymb}
\usepackage{subcaption}
\usepackage{wrapfig}
\usepackage{import}
\usepackage{tcolorbox}

\newtheorem{definition}{Definition}
\newtheorem{theorem}{Theorem}
\newtheorem{remark}{Remark}
\newtheorem{proposition}{Proposition}
\newtheorem{lemma}{Lemma}
\newtheorem{corollary}{Corollary}

\def\enc{\mathbf{u}_\textrm{e}}
\def\dec{\mathbf{u}_\textrm{d}}
\def\encs{u_\textrm{e}}
\def\decs{u_\textrm{d}}

\def\decsbar{\bar{u}_\textrm{d}}

\def\enclamb{\lambda_\textrm{e}}
\def\declamb{\lambda_\textrm{d}}
\def\lenc{\mathcal{L}_\textrm{enc}}
\def\ldec{\mathcal{L}_\textrm{dec}}

\def\func{\mathbf{f}}

\def\fhat{\mathbf{\hat{f}}}

\newcommand{\norm}[1]{\left\lVert#1\right\rVert}

\newcommand{\hil}[1]{\mathcal{H}^{#1}\left(\Omega\right)}
\newcommand{\hilR}[2]{\mathcal{H}^{#1}\left(#2\right)}

\newcommand{\hilz}[1]{\mathcal{H}_0^{#1}\left(\Omega\right)}
\newcommand{\hilmz}[1]{\mathcal{H}_0^{#1}\left(\Omega;\mathbb{R}^M\right)}
\newcommand{\hilm}[2]{\mathcal{H}^{#1}_{#2}\left(\Omega\right)}
\newcommand{\hiltilde}[1]{\widetilde{\mathcal{H}}^{#1}\left(\Omega; \mathbb{R}\right)}
\newcommand{\hilmtilde}[2]{\widetilde{\mathcal{H}}^{#1}\left(\Omega; \mathbb{R}^#2\right)}

\newcommand{\lp}[2]{{L}^{#1}\left(#2\right)}
\newcommand{\lpm}[2]{{L}^{#1}\left(#2; \mathbb{R}^M\right)}
\newcommand{\lploc}[2]{{L}_{\textrm{loc}}^{#1}\left(#2; \mathbb{R}\right)}
\newcommand{\lpmloc}[2]{{L}_{\textrm{loc}}^{#1}\left(#2; \mathbb{R}^M\right)}

\newcommand{\smp}{\mathbb{W}^{m,p}\left(\Omega; \mathbb{R}^M\right)}
\newcommand{\smpz}{\mathbb{W}_0^{m,p}\left(\Omega; \mathbb{R}^M\right)}
\newcommand{\smpeq}{\widetilde{\mathbb{W}}^{m,p}\left(\Omega; \mathbb{R}^M\right)}
\newcommand{\smploc}{\mathbb{W}_{\textrm{loc}}^{m,p}\left(\Omega; \mathbb{R}^M\right)}
\newcommand{\soblocm}[2]{\mathbb{W}_{\textrm{loc}}^{#1,#2}\left(\Omega; \mathbb{R}^M\right)}

\newcommand{\sob}[2]{\mathbb{W}^{#1,#2}\left(\Omega; \mathbb{R}\right)}

\newcommand{\sobeq}[2]{\widetilde{\mathbb{W}}^{#1,#2}\left(\Omega; \mathbb{R}\right)}
\newcommand{\sobm}[2]{\mathbb{W}^{#1,#2}\left(\Omega; \mathbb{R}^M\right)}
\newcommand{\sobmz}[2]{\mathbb{W}_0^{#1,#2}\left(\Omega; \mathbb{R}^M\right)}
\newcommand{\sobmeq}[2]{\widetilde{\mathbb{W}}^{#1,#2}\left(\Omega; \mathbb{R}^M\right)}

\newcommand{\spline}{\mathbf{S}_{\lambda,n,m}}

\newcommand{\lec}[2]{\stackrel{\text{#1}}{#2}}
\begin{document}
\title{General Coded Computing: Adversarial Settings} 

% %%% Single author, or several authors with same affiliation:
% \author{%
%  \IEEEauthorblockN{Author 1 and Author 2}
% \IEEEauthorblockA{Department of Statistics and Data Science\\
%                    University 1\\
 %                   City 1\\
  %                  Email: author1@university1.edu}% }

%%% Several authors with up to three affiliations:
\author{%
  \IEEEauthorblockN{Parsa Moradi}
  \IEEEauthorblockA{
                    University of Minnesota, Twin Cities\\
                    Minneapolis, MN, USA}
                    \and
  \IEEEauthorblockN{Hanzaleh Akbarinodehi}
  \IEEEauthorblockA{
                    University of Minnesota, Twin Cities\\
                    Minneapolis, MN, USA}
                    \and
    \IEEEauthorblockN{Mohammad Ali Maddah-Ali}
  \IEEEauthorblockA{
                    University of Minnesota, Twin Cities\\
                    Minneapolis, MN, USA}
}
\maketitle

%%%%%%
%% Abstract: 
%% If your paper is eligible for the student paper award, please add
%% the comment "THIS PAPER IS ELIGIBLE FOR THE STUDENT PAPER
%% AWARD." as a first line in the abstract. 
%% For the final version of the accepted paper, please do not forget
%% to remove this comment!
%%

\begin{abstract}
Conventional coded computing frameworks are predominantly tailored for structured computations, such as matrix multiplication and polynomial evaluation. Such tasks allow the reuse of tools and techniques from algebraic coding theory to improve the reliability of distributed systems in the presence of stragglers and adversarial servers.

This paper lays the foundation for general coded computing, which extends the applicability of coded computing to handle a wide class of computations. In addition, it particularly addresses the challenging problem of managing adversarial servers.
We demonstrate that, in the proposed scheme, for a system with $N$ servers, where $\mathcal{O}(N^a)$,  $a \in [0,1)$, are adversarial, the supremum of the average approximation error over all adversarial strategies decays at a rate of $N^{\frac{6}{5}(a-1)}$, under minimal assumptions on the computing tasks. Furthermore, we show that within a general framework, the proposed scheme achieves optimal adversarial robustness, in terms of maximum number of adversarial servers it can tolerate. This marks a significant step toward practical and reliable general coded computing. Implementation results further validate the effectiveness of the proposed method in handling various  computations, including inference in deep neural networks.
\end{abstract}
% \begin{abstract}
%     Coded computing has demonstrated remarkable potential in addressing critical challenges in distributed systems, such as the presence of slow (straggler) or adversarial servers. However, most existing frameworks are tailored for highly structured functions, such as polynomials, and focused on exact recovery. Therefore, their process fail completely if the number of available servers falls below a certain threshold.
%     While recent work has proposed coded computing schemes for general function approximation with straggler resiliency guarantees, the adversarial setting—where some servers arbitrarily alter their computation results—remains unexplored.
%     In this paper, we introduce, for the first time (to the best of our knowledge), an adversarial robust \emph{General Coded Computing} framework for general function approximation. Our proposed scheme enables robust approximate computing in the presence of adversarial servers. Specifically, we show that for a system with $N$ servers, of which $N^a$ are adversarial, the approximation error decays at a rate of $N^{-\frac{6}{5}(a-1)}$, under minimal structural assumptions on the computing function. Furthermore, the proposed scheme achieves optimal adversarial robustness, in terms of maximum number of adversarial servers it can tolerate, representing a significant step toward practical and reliable general coded computing.
% \end{abstract}
\section{Introduction}\label{sec:introduction}
Coded computing has demonstrated remarkable potential in addressing critical challenges in distributed systems, such as the presence of slow (straggler) or adversarial servers \cite{jahani2022berrut,yu2020straggler,karakus2017straggler,soleymani2022approxifer, yu2019lagrange, moradicoded}.  
In this approach, each server, instead of processing raw input, it process a combination of input data. The coded redundancy in the inputs, if designed properly, induces redundancy in the computation results. This redundancy enables the system to tolerate missing results from stragglers and enhances its robustness and reliability against adversarial attacks \cite{yu2019lagrange, das2019randoms, karakus2019redundancy, karakus2017straggler, yu2020straggler, entangle}.

% ised
% This scheme consists of one master node and set of servers, referred to as worker nodes. The master node aims to compute a specific function on a given set of input data, using the computing resource of the worker nodes. Firstly, the master node sends a combination of input data to each worker (encoding). Then, the worker nodes apply the computing function on their given input and send the results, back to the master node (computing). Finally, the master node recover the output of function on the original data points from the workers results (decoding).

Building on the insights and techniques developed in algebraic coding theory, coded computing techniques are primarily designed for structured computations, such as matrix multiplication \cite{yu2017polynomial,RamamoorthyConv,gupta2018oversketch,jahani2018codedsketch, opt-recovery, short} and polynomial evaluation \cite{yu2019lagrange}.
These structures facilitate the design of the encoding process such that, when integrated with the computation task, the results follow a desired form.
For example, the encoding process is designed so that, after computation, the results lie on a polynomial curve~\cite{yu2017polynomial,yu2019lagrange}. Thus, the theoretical guarantees of polynomial-based codes, such as Reed-Solomon codes, not only hold but also facilitate the use of efficient decoding algorithms, such as the Berlekamp–Welch algorithm \cite{blahut2008algebraic}.

% Coded computing has demonstrated significant potential in addressing critical challenges associated with the distributed computation of \emph{highly structured} functions, such as polynomials \cite{yu2017polynomial,yu2019lagrange} and matrix multiplications \cite{}. These challenges include the presence of stragglers (faulty or slow workers) \cite{} and Byzantine workers, who may act maliciously either intentionally or unintentionally \cite{}.

Besides its limited capacity to handle a wide range of computation tasks, conventional coded computing has several additional shortcomings.
These include numerical instability when performing computations over real numbers, inefficiency in handling or leveraging approximate computing, and high \emph{recovery thresholds}, which require a large number of non-failed or non-straggling servers for successful computation. Extensive research efforts have been undertaken to address these challenges, with each approach aiming to resolve specific issues \cite{fahimappx, RamamoorthyConv, soleymani2020analog, ramamoorthy2021numerically, AnooshehRobust,overSketch,jahani2018codedsketch}. However, to effectively address these challenges, it is necessary to revisit the concept of coded computing, starting with its foundational design principles.

In this paper, we lay the foundation of \emph{General Coded Computing}, which is not solely rooted in the principles of error correction coding but also draws on concepts from approximation and learning theory. The effectiveness of this new approach is validated through numerical evaluations and theoretical guarantees. Our focus is particularly on addressing adversarial behavior, which serves as a significant complement to our team's earlier efforts in mitigating the effects of stragglers for general computations~\cite{jahani2022berrut, moradicoded}. This task is especially challenging, as managing adversarial behavior has not been sufficiently addressed within the areas of approximation and learning theory.

Consider the task of computing $\func(\mathbf{x}_1), \ldots, \func(\mathbf{x}_K)$, for some general function $\func: \mathbb{R}^d \to [-M, M]^m$, and input data points $\mathbf{x}_1, \ldots, \mathbf{x}_K \in \mathbb{R}^d$, $K, m \in \mathbb{N}$ and $M \in \mathbf{R}_+$. The master node embeds the input data points into an encoding function 
$\enc(.)$, and distributes the samples of $\enc(.)$, as the coded inputs, to worker nodes, which compute 
$\func(.)$ of these coded symbols. Some of these servers are adversaries.  At the decoding side, we recover a decoding function $\dec(.)$ from the results of the worker nodes. However, instead of selecting the encoder and decoder from polynomials, like Reed-Solomon and Lagrange coding,  $\enc(.)$ and  $\dec(.)$ belongs to the set of functions with bounded first and second-order derivatives, i.e., the Second-order Sobolev spaces. We show that while such a smoothness condition is sufficiently restrictive to introduce redundancy in the coded symbols and mitigate the impact of adversarial worker nodes, it remains flexible enough to accommodate general computation tasks, with minimal assumptions. We use the properties of Sobolev spaces, as Reproducing Hilbert Kernel Spaces (RHKS), to efficiently form the encoding and decoding functions.  We prove that in a system with $N$ worker nodes, where $\mathcal{O}(N^a)$, $a \in [0,1)$, are adversarial, the maximum approximation error of the proposed scheme, taken over all possible adversarial actions, decays at a rate of $N^{\frac{6}{5}(a-1)}$,  under the mild assumption that the computing function belongs to the second-order Sobolev space. Additionally, we prove that within this framework, the proposed scheme is optimal in the sense that no other scheme can tolerate a greater number of adversarial worker nodes. Implementation results further validate the effectiveness of the proposed method in handling complex computations, such as inference in deep neural networks. 

The paper is organized as follows: Section~\ref{sec:prelim} introduces the problem formulation. Section~\ref{sec:scheme} provides a detailed description of the proposed scheme. Section~\ref{sec:theory} outlines the theoretical results, and Section~\ref{sec:exp_res} presents the experimental results. Lastly, Section~\ref{sec:proofs} includes proof sketches for the main theorems.

{\bf Notations:}
In this paper, we denote the set $\{1, 2, \ldots, n\}$ by $[n]$ and the size of a set $S$ by $|S|$. Bold letters represent vectors and matrices, while the $i$-th component of a vector-valued function $\mathbf{f}$ is written as $f_i(\cdot)$. Scalar function derivatives are denoted by $f'$, $f''$, and $f^{(k)}$ for first, second, and $k$-th orders. The space $\hilm{2}{m}$ refers to the RKHS of second-order Sobolev functions for vector-valued functions of dimension $m$ on $\Omega$, with $\hilR{2}{\Omega}$ used for the one-dimensional case.

\section{Problem Formulation}\label{sec:prelim}

%\subsection{Problem Setting}\label{sec:prob_set}

Let $\func: \mathbb{R}^d \to [-M, M]^m$, be a function where $\func(\mathbf{x}) \triangleq [f_1(\mathbf{x}), f_2(\mathbf{x}), \dots, f_m(\mathbf{x})]$, and $ f_i: \mathbb{R}^d \to [-M, M]$ for some $m,d \in \mathbb{N}$, $M \in \mathbb{R}$, and $i\in [m]$. The function $\func(\cdot)$ can represent an arbitrary function, ranging from a simple scalar function to a complex deep neural network.

Consider a system consisting of a master node and $N \in \mathbb{N}$ worker nodes. The master node is interested in the approximate results of computing $\{\func(\mathbf{x}_k)\}_{k=1}^K$ in a distributed setting, leveraging the set of $N$ worker nodes, for $K \in \mathbb{N}$ data points $\{\mathbf{x}_k\}_{k=1}^K$, where each $\mathbf{x}_k \in \mathbb{R}^d$.
%  Consider a setup with a master node and $N \in \mathbb{N}$ worker nodes. The master node's goal is to compute the set $\{\func(\mathbf{x}_k)\}_{k=1}^K$ for $K$ data points $\{\mathbf{x}_k\}_{k=1}^K$, where each $\mathbf{x}_k \in \mathbb{R}^d$, for some $d,K \in \mathbb{N}$.  The function $\func: \mathbb{R}^d \to [-M, M]^m, \quad \text{where } \func(x) = [f_1(x), f_2(x), \dots, f_m(x)] \text{ and } f_i(x) \in [-M, M]$ for $i\in [N]$.  The only assumptions for $\func(\cdot)$ are Lipschitz continuity and a bounded second derivative. (\han{Change it to: We assume that $\func(.)$ is $L$-Lipchitz, that is for any $x_1, x_2 \in ....$, we have
% \begin{align}
%     ...
% \end{align}.
% In addition, we assume that for any $x \in ...$, we have $|\func''(x)| \leq C$, for some $C \in \mathbb{R}$.} ) 
The set of worker nodes is partitioned into two sets: (1) the set of \emph{honest nodes}, denoted by \( \mathcal{D} \), which strictly execute any assigned task, (ii) the set of \emph{adversarial nodes}, denoted by \( \mathcal{B} \), which  produce arbitrary outputs. We assume that \( |\mathcal{B}| \leq \gamma \), for some \( \gamma \leq N \). 
% \( \mathcal{D} \cap \mathcal{B} = \emptyset \), \( \mathcal{D} \cup \mathcal{B} = [N] \), and
We assume that all adversarial worker nodes are controlled by a single entity, referred to as the \emph{adversary}. 
Throughout this paper, we adopt the following three-step framework for coded computing:
\begin{enumerate}
    \item \textbf{Encoding:} The master node forms an encoding function $\enc: \mathbb{R} \to \mathbb{R}^d$, such that $\mathbf{x}_k \approx \enc(\alpha_k)$, where $\alpha_1 < \alpha_2 < \dots < \alpha_K \in \Omega \subset \mathbb{R}$ are fixed and predefined constants. The accuracy of the above approximation is determined during the encoder design process, as will be described later. It then evaluates $\enc(\cdot)$ on another set of fixed points $\{\beta_n\}_{n=1}^N$, where $\beta_1 < \beta_2 < \dots < \beta_N \in \Omega \subset \mathbb{R}$. For each $n \in [N]$, the master node computes the coded data point $\Tilde{\mathbf{x}}_n = \enc(\beta_n) \in \mathbb{R}^d$ and sends it to worker node $n$, $n \in [N]$.  Note that each coded data point is a combination of all the original points $\{\mathbf{x}_k\}_{k=1}^K$.

        \item \textbf{Computing:} Every honest worker node $n \in \mathcal{D}$ computes $\func(\Tilde{\mathbf{x}}_n) = \func(\enc(\beta_n))$ and sends it back to the master node. On the other hand, an adversarial worker node $n \in \mathcal{B}$ may send arbitrary values, still within the acceptable range $[-M, M]^m$ in an attempt to avoid detection and exclusion.
        Thus, in summary,  worker node $n$ sends $\overline{\mathbf{y}}_n$ to the master node as follows: 
%         $
% \overline{\mathbf{y}}_n =  
%     \func(\Tilde{\mathbf{x}}_n)$ if  $n \in \mathcal{D}$ and $\overline{\mathbf{y}}_n =\star \in [-M, M]^m$ if $n \in \mathcal{B}$.
        % $\overline{\mathbf{y}}_n = \func(\Tilde{\mathbf{x}}_n)$, for $n \in \mathcal{D}$, and $\overline{\mathbf{y}}_n = \star \in [-M, M]^m$, otherwise. 

 % \begin{comment}   
        $$
    \overline{\mathbf{y}}_n = \begin{cases} 
    \func(\Tilde{\mathbf{x}}_n) & \text{if } n \in \mathcal{D}, \\
    \star \in [-M, M]^m & \text{if } n \in \mathcal{B}.
    \end{cases}
    $$
    % \end{comment}
    % which is the output range of the computing function $\func(\cdot)$.
   
    \item \textbf{Decoding:} Upon receiving the results from the worker nodes, the master node forms a decoder function $\dec: \mathbb{R} \to \mathbb{R}^m$, such that $\dec(\beta_n)\approx \overline{\mathbf{y}}_n$, for $n \in \mathcal{D}$.  Of course,  the master node is unaware of the set $\mathcal{D}$,  which is one of the reasons that determining $\dec$ is challenging.
 It then computes $\hat{\func}(\mathbf{x}_k) := \dec(\alpha_k)$ as an approximation of $\func(\mathbf{x}_k)$ for $k \in [K]$. If the master node can effectively mitigate the influence of adversarial worker nodes, it expects the approximation $\dec(\alpha_k) \approx \func(\enc(\alpha_k)) \approx \func(\mathbf{x}_k)$ to hold.
\end{enumerate}
 
Recall that $|\mathcal{B}| \leqslant \gamma$. Let $\mathcal{A}_\gamma$ denote the set of all possible strategies that the adversary can employ.  The adversary chooses its attack to maximize the approximation error of the computation, leveraging full knowledge of the input data points, the computation function $\func(\cdot)$, $\{\alpha_k\}_{k=1}^K$,  $\{\beta_n\}_{n=1}^N$, and  and the scheme used to design the encoder and decoder functions $\enc(\cdot)$ and $\dec(\cdot)$. The effectiveness of the proposed coded computing scheme is evaluated using the \emph{average approximation error}, defined as:
\begin{align}\label{eq:obj}
    \mathcal{R}(\fhat) := \mathop{\sup}_{\mathcal{A}_\gamma} \frac{1}{K} \sum^K_{k=1} \norm{\fhat(\mathbf{x}_k) - \func(\mathbf{x}_k)}^2_2.
\end{align}
The objective is then to find $\enc \in \hilm{2}{d}$ and $\dec\in \hilm{2}{m}$ such that  minimize $\mathcal{R}(\fhat)$ in \eqref{eq:obj}. The proposed approach raises several important questions: 
(i) Does constraining the encoder and decoder to second-order Sobolev spaces (i.e., functions with bounded first and second derivative norms) impose sufficient restrictions to ensure robustness of the distributed system against adversarial attacks? (ii) If so, what is the maximum number of adversaries the scheme can tolerate while maintaining its effectiveness?
(iii)  How fast does the average estimation error converge to zero as a function of $N$ and $\gamma$?

% We are interested to answer the following questions: (1) Is this scheme able to protect the system against adversaries? In other words, does limiting the encoder and decoder to the second-order Sobolev spaces (functions with bounded first and second order derivatives) is enough to smooth out adversarial results? (2) if yes, how many adversary can this scheme handle? 
% (3) If this scheme is effective, how fast the average estimation error converges to zero, as a function of $N$ and the number of adversarial servers? 

We note that minimizing~\eqref{eq:obj}, subject to $\enc \in \hilm{2}{d}$ and $\dec\in \hilm{2}{m}$,  is a challenging task for several reasons. First, the assumptions on the function $\func(\cdot)$ are minimal; we only require that $\func(\cdot)$ has bounded first and second derivative norms.  
Second, designing these functions involves addressing the complexities of infinite-dimensional optimization, as both the encoder and decoder are elements of Sobolev spaces. Third, optimizing the encoder and decoder cannot be treated independently. The objective function $\mathcal{R}(\fhat)$ in~\eqref{eq:obj} arises from the composition of three functions: the encoding function, the computation task, and the decoding function, necessitating a joint optimization approach. 
Lastly, the supremum in \eqref{eq:obj} is taken over $\mathcal{A}_\gamma$, representing all possible adversarial behaviors. This requirement to optimize against worst-case adversarial strategies significantly increases the complexity of the problem.

% This optimization problem is highly challenging for several reasons. First, the assumptions on the function $\func$ are minimal; we only require that $\func$ has bounded first and second derivative norms. This lack of additional structural assumptions makes the problem particularly difficult to address. Second, the optimization spans an infinite-dimensional space, as $\enc(\cdot)$ and $\dec(\cdot)$ belong to Sobolev spaces of functions. Designing these functions requires overcoming the complexities associated with this infinite-dimensional optimization framework. Third, $\enc(\cdot)$ and $\dec(\cdot)$ are interdependent and must satisfy specific compatibility conditions, as shown later in \eqref{eq:decompose}. These conditions impose additional constraints on their design. Lastly, the supremum in \eqref{eq:obj} is taken over $\mathcal{A}_\gamma$, representing all possible adversarial strategies. Optimizing against this worst-case adversarial behavior adds further complexity to the problem.

% The objective is to determine $\enc(\cdot)$ and $\dec(\cdot)$ that minimize \eqref{eq:obj}. This presents an infinite-dimensional optimization problem with a complex encoder-decoder relationship (see \eqref{eq:decompose}), making it highly challenging. We argue that, under mild assumptions on the encoder and decoder functions, where $\enc(\cdot) \in \hilm{2}{d}$ and $\dec(\cdot) \in \hilm{2}{m}$, there exist a scheme that provably mitigates the maximum possible number of adversaries.

\section{Proposed Scheme}\label{sec:scheme}
%We assume that $\enc \in \hilm{2}{d}$ and $\dec \in \hilm{2}{m}$.$ 
% The objective is to design the encoding and decoding functions in a way that minimizes the approximation error expressed in \eqref{eq:obj}. 
Considering the difficulty of optimizing $\mathcal{R}(\fhat)$ in~\eqref{eq:obj},  we first establish an upper bound for $\mathcal{R}(\fhat)$ as the sum of two terms:
\begin{alignat}{2}
    \label{eq:decompose}
    \mathcal{R}(\fhat) &\lec{(a)}{=} 
     \sup_{\mathcal{A}_\gamma} \frac{1}{K}
     \sum^K_{k=1} \|\dec(\alpha_k)- \func(\enc(\alpha_k))  + \func(\enc(\alpha_k)) - \func(\mathbf{x}_k)\|^2_2 \nonumber \\
    &\lec{(b)}{\leqslant}  \underbrace{\mathop{\sup}_{\mathcal{A}_\gamma} \frac{2}{K} \sum^K_{k=1} \norm{\dec \left(\alpha_k\right) - \func \left(\enc\left(\alpha_k\right)\right)}^2_2}_{\ldec(\hat{\func})}+  \underbrace{\frac{2}{K} \sum^K_{k=1} \norm{\func(\enc(\alpha_k)) - \func(\mathbf{x}_k)}^2_2}_{\lenc(
 \hat{\func})},
\end{alignat}
where (a) follows from adding and subtracting the term $\func(\enc(\alpha_k))$, and (b) is due to the AM-GM inequality.

The first term in \eqref{eq:decompose}, $\ldec(\fhat)$, represents the error of the decoder function in predicting the target function $\func(\enc(\cdot))$ at the points $\{\alpha_k\}_{k=1}^K$, which differ from its training points $\{\beta_n\}_{n=1}^N$, for which the master node has partial access (except from adversarial inputs) to the value of $\func(\enc(\cdot))$. 
Thus, in the terminology of learning theory, it would represents the generalization error of the decoder function.
In contrast, the second term, $\lenc(\fhat)$, serves as a proxy for the encoder's training error, as it captures the encoder's error on $\{\alpha_k\}_{k=1}^K$, after passing through $\func(\cdot)$. These two terms are not independent, as the target function of the decoder is $\func(\enc(\cdot))$, which depends on the encoder function. Consequently, \eqref{eq:decompose} reveals an inherent interplay between the encoder and decoder functions, further complicating the optimization of \eqref{eq:obj}.

\subsection{Decoder Design}\label{sec:dec_design}
Recall that it is important for the decoder function to achieve both a low generalization error and high robustness to adversarial inputs. Therefore, we propose to find $\dec$, using the following regularized optimization:
\begin{align}\label{eq:decoder_opt_adv}
\dec=\underset{\mathbf{u} \in \hilm{2}{m}}{\operatorname{argmin}} \frac{1}{N} \sum^N_{n=1}\norm{\mathbf{u}\left(\beta_n\right)- \mathbf{\overline{y}}_n}_2^2+\declamb \norm{\mathbf{u}''}^2_{\lp{2}{\Omega}},
\end{align}
where $\norm{\mathbf{u}''}_{\lp{2}{\Omega}}$ denotes the $\ell_2$ norm of the vector-valued function $\mathbf{u}''(\cdot)$. The first term in \eqref{eq:decoder_opt_adv} represents the mean squared error, which \emph{controls the accuracy} of the decoder function. The second term serves as a regularization term, promoting smoothness of the decoder function over the interval $\Omega$ and thereby \emph{controlling generalization and robustness}. The parameter $\declamb$, known as the smoothing parameter, determines the relative weighting of these two terms.  
% As $\declamb \to \infty$, the solution to \eqref{eq:decoder_opt_adv} converges to a vector-valued function whose elements are second-degree polynomials. Conversely, as $\declamb \to 0$, the decoder function tends to overfit the training data.

Using the representer theorem \cite{scholkopf2001generalized}, it can be shown that the optimal decoder $\mathbf{u}^*(\cdot) := [u_1^*(\cdot), \dots, u_m^*(\cdot)]$ in \eqref{eq:decoder_opt_adv} can be expressed such that each element $u_i^*(\cdot)$ is a linear combination of the basis functions $\{\phi(\cdot, \beta_n)\}_{n=1}^N$, where $\phi(\cdot, \cdot)$ denotes the kernel function associated with $\hil{2}$.
The solution to \eqref{eq:decoder_opt_adv} belongs to a class of spline functions known as \emph{second-order smoothing splines} \cite{wahba1975smoothing, wahba1990spline}. The coefficient vectors for $\mathbf{u}^*_i$, corresponding to $\{\phi(\cdot, \beta_n)\}_{n=1}^N$, can be efficiently computed in $\mathcal{O}(Nm)$ by utilizing B-spline basis functions \cite{wahba1990spline, eilers1996flexible}.
% and solving a quadratic optimization problem \cite{wahba1975smoothing}

In practice, the hyper-parameter $\declamb$ in \eqref{eq:decoder_opt_adv} is typically determined using cross-validation \cite{wahba1975smoothing, wahba1990spline}. However, we perform a theoretical analysis to determine the optimal value of $\declamb$ that minimizes an upper bound for \eqref{eq:obj}. Specifically, we show that if the maximum number of adversarial worker nodes is $\mathcal{O}(N^a)$ with $a \in [0,1)$, then the optimal value of $\declamb$ is $\mathcal{O}(N^{\frac{8}{5}(a-1)})$.

\subsection{Encoder Design}\label{sec:enc_des}
We will show that $\ldec$ can be bounded above by an expression, which can serve as  a regularization term in searching for the encoder function (see Theorem~\ref{th:enc_bound}). Consequently, the encoder function can be determined by the following optimization: 
\begin{align*}
    \enc=\underset{\mathbf{u} \in \hilm{2}{d}}{\operatorname{argmin}} \frac{C}{K} \sum^K_{k=1}\norm{\mathbf{u}\left(\alpha_k\right)- \mathbf{\mathbf{x}}_k}^2+\enclamb \psi\left(\norm{\mathbf{u}}^2_{\hilm{2}{d}}\right),
\end{align*}
where where $\enclamb$ is a function of $(\declamb, N, M, \gamma)$, $C$ depends on the Lipschitz constant of the computing function, and $\psi(\cdot)$ is a monotonically increasing function in $[0, \infty)$. Applying the representer theorem \cite{scholkopf2001generalized}, one can conclude that, similar to the decoder function, the encoder function can also be represented as a linear combination of $\{\phi(\cdot, \alpha_k)\}_{k=1}^K$.

\section{Theoretical Guarantees}\label{sec:theory}
In this section, we analyze the proposed scheme from theoretical viewpoint. For simplicity of exposition, in this conference paper, we focus on one-dimensional case in which $f:\mathbb{R} \to \mathbb{R}$. Also, we assume that $\Omega = [0,1].$

First, we present the following theorem, which demonstrates that if $\gamma = \mathcal{O}(N)$, no encoder and decoder functions can achieve adversarial robustness:
\begin{theorem}{\normalfont(Impossibility Result)}\label{th:impos}
    In the proposed  coded computing framework,  assume $\gamma =  \mu N$ for some $\mu \in (0, 1)$. Then, there exists some $f$ with bounded first and second derivative, for which there is no $\encs, \decs \in \hil{2}$ such that:
    $
        \lim_{N \to \infty} \mathcal{R}(\hat{f}) = 0.
    $
\end{theorem}
Next, we demonstrate  that the proposed scheme achieves adversarial robustness when $\gamma = o(N)$. More precisely, the next theorem establishes an upper bound for the approximation error as a function of $\declamb$, $N$, $M$, and $\gamma$, providing insights into the adversarial robustness of the proposed scheme. Subsequently, we utilize this theorem to design the encoder.
\begin{theorem}\label{th:dec_bound}
    Consider the proposed coded computing scheme with $N$ worker nodes and 
    at most $\gamma$ colluding worker nodes. Assume decoder points are equidistant, i.e. $\beta_i = \frac{i}{N}$ for $i \in [N]$ and $C_{\lambda}N^{-4} < \declamb \leqslant 1$ for constant $C_{\lambda}>0$. If $\norm{f'}_{\lp{\infty}{\Omega}}\leqslant\nu$  and $\norm{f''}_{\lp{\infty}{\Omega}} \leqslant \eta$, then:
\begin{align}\label{eq:th1}
    \mathcal{R}(\hat{f}) \leqslant 
    &C_1 \frac{M^2\gamma^2}{N^4} + C_2\frac{M^2\gamma^2}{N^2}\declamb^{-\frac{1}{2}}\left( e^{\left({\sqrt{2}}\declamb^{-\frac{1}{4}}\right)} + C_3\right)
    \nonumber \\
    &+\left(C_4\declamb^\frac{3}{4} + C_5N^{-3}\right) \norm{(f \circ \encs)''}^{2}_{\lp{2}{\Omega}} 
    \nonumber \\
    &+ \frac{2\nu^2}{K} \sum^K_{k=1} (\encs(\alpha_k) - x_k)^2,
\end{align}
where $C_i$s are constants.
\end{theorem}
Theorem~\ref{th:dec_bound} holds for all $\encs \in \hilR{2}{\Omega}$. However, these bounds do not directly offer a method for designing $\encs(\cdot)$, mainly because of the inherent complexity of the composition function $f \circ \encs$. Inspired by \cite{moradicoded}, the following theorem establishes an upper bound for the approximation error in terms of $\norm{\encs}^2_{\hilR{2}{\Omega}}$.
\begin{theorem}\label{th:enc_bound} 
Consider the proposed scheme. Suppose that the computing function satisfies $\norm{f''}_{\lp{\infty}{\Omega}} \leqslant \eta, \norm{f'}_{\lp{\infty}{\Omega}} \leqslant \nu$. Then, there exists $\enclamb > 0$, which depends on $(\declamb, N, M, \gamma, \nu, \eta)$, and a monotonically increasing function $\psi:\mathbb{R}^+ \to \mathbb{R}^+$, independent of the scheme's parameters, such that
\begin{align}\label{eq:optimcal_encdoer} 
\mathcal{R}(\hat{f}) \leqslant \frac{2\nu^2}{K} \sum_{k=1}^K (\encs(\alpha_k) - x_k)^2 + \enclamb \cdot \psi\left(\norm{\encs}^2_{\hilR{2}{\Omega}}\right).
\end{align}
\end{theorem}
The above theorem provides the foundation for the proposed approach outlined in Section~\ref{sec:enc_des} to design the encoder function. Finally, the following corollary demonstrate the convergence rate of the proposed scheme:

\begin{corollary} {\normalfont(Convergence Rate)}\label{cor:conv_rate}
     Let $\gamma = \mathcal{O}(N^a)$ for $a \in [0, 1)$ and $\declamb = \mathcal{O}(N^{\frac{8}{5}(a-1)})$. Then, the average approximation error $\mathcal{R}(\hat{f})$ of the proposed scheme converges to zero at a rate of at least
     $\mathcal{O}(N^{\frac{6}{5}(a-1)})$.
    % the term \han{$\ldec$} achieves a convergence rate of $\mathcal{O}(N^{\frac{6}{5}(a-1)})$ (\han{Mention the encoder is not dependent to $n$, and thus the composition of the function and encoder is not dependent to $n$}). Additionally, the encoder function can be designed (\han{where?}) such that the term $\frac{2\nu^2}{K} \sum_{k=1}^K (\encs(\alpha_k) - x_k)^2$ vanishes at the maximum rate of $\mathcal{O}(N^{-\frac{6}{5}(a-1)})$ (\han{I think this part is misleading the reader. The encoder can be designed to over fit, and thus the term $\frac{2\nu^2}{K} \sum_{k=1}^K (\encs(\alpha_k) - x_k)^2$ would be zero. No rate needed for this part!}). As a result, $\mathcal{R}(\hat{f})$ also converges to zero at the same rate.
\end{corollary}
\begin{remark}{(\normalfont Optimality)
    According to Theorem~\ref{th:impos} and Corollary~\ref{cor:conv_rate}, the proposed scheme achieves optimal robustness in terms of the number of adversarial worker nodes.}
    % , if $\gamma = o(N)$, the proposed scheme demonstrates adversarial robustness, implying that as the number of workers increases, the approximation error diminishes to zero.
\end{remark}

% Note that the space of smoothing spline functions (solutions to objective functions like \eqref{eq:decoder_opt}) is a subset of this class of encoder functions.

Finally, the following theorem states that, under mild assumptions on $\norm{\encs}^2_{\hilR{2}{\Omega}}$, an upper bound for the approximation error can be established, leading us to the design of the encoder function:
\begin{theorem}\label{th:letcc_enc_ss}
    Consider the proposed scheme under the same assumptions as Theorems~\ref{th:dec_bound} and \ref{th:enc_bound}. There exists $ E > 0 $, depending only on the encoder regression points $\{\alpha_k\}_{k=1}^K $ and the input data points $ \{x_k\}_{k=1}^K $, such that:
    \begin{itemize}
    \item[(i)] If $\norm{\encs}^2_{\hilR{2}{\Omega}} \leqslant E$, then:
    \begin{align}\label{eq:encoder_linear_upbound}
        \mathcal{R}(\hat{f}) &\leqslant \frac{2\nu^2}{K} \sum_{k=1}^K (\encs(\alpha_k) - x_k)^2 + \enclamb \left(D_1(E) + D_2(E) \int_{\Omega}\left(\encs''(t)\right)^2\,dt\right),
    \end{align}
    where $ D_1(E) $ and $ D_2(E)$  are functions of $ E $.
    \item[(ii)] If  $u_e^*(\cdot)$ is the  minimizer of \eqref{eq:encoder_linear_upbound}, then $\norm{u_e^*}^2_{\hilR{2}{\Omega}} \leqslant E$.
    \end{itemize}
\end{theorem}
This result demonstrates that the optimal encoder function, which minimizes the upper bound in Theorem~\ref{th:letcc_enc_ss}, is a second-order smoothing spline, offering higher computational efficiency.
\begin{figure*}[t]
     \centering
     \begin{subfigure}[b]{0.495\textwidth}
         \centering
         \includegraphics[width=0.95\textwidth, height=0.6\textwidth]{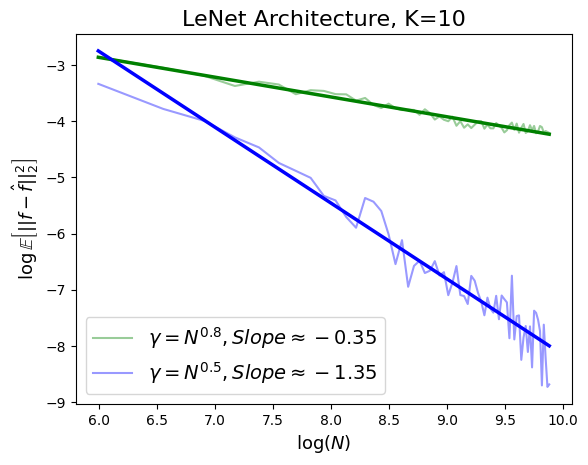}
         \label{fig:rate_lenet}
     \end{subfigure}
     \hfill
     \begin{subfigure}[b]{0.495\textwidth}
         \centering
         \includegraphics[width=0.95\textwidth, height=0.6\textwidth]{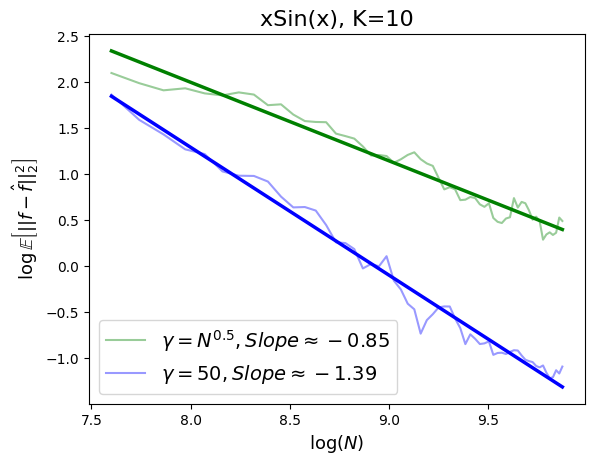}
         \label{fig:rate_xsinx}
     \end{subfigure}
     %\captionsetup{font=small}
     \caption{Log-log plot illustrating the convergence rates of approximation error for the function $f(x) = x\sin(x)$ and LeNet5 network under various number of adversarial worker nodes.}
     \label{fig:rate_all}
\end{figure*}

\section{Experimental Results}\label{sec:exp_res}
In this section, we evaluate the performance of the proposed scheme in different adversarial settings. Two types of computing functions are considered: a one-dimensional function $f_1(x) = x\sin(x)$ and the LeNet5 neural network \cite{lecun1998gradient}, a high-dimensional function $\mathbf{f}_2: \mathbb{R}^{1024} \to \mathbb{R}^{10}$ trained for handwritten image classification.
Equidistant points are used for both the encoder and decoder. The scheme's approximation error is quantified as empirical approximation of  $\mathrm{E}_{\mathbf{x} \sim \mathcal{X}}[\mathcal{R}(\hat{\mathbf{f}})]$, averaged over 20 repetitions for each $N$.

The adversary's strategy is defined as modifying the values of $f(\encs(\beta_i))$ for all $\beta_i$ near  $\{\alpha_k\}_{k=1}^K$. Specifically, the adversary changes $\frac{\gamma}{K}$ values around each $\alpha_i$ for $i \in [K]$ to maximum acceptable value $M$.
As shown in Figure~\ref{fig:rate_all}, for $\gamma = N^{0.8}$, the average approximation error for LeNet5 network converges at a rate of $N^{-0.35}$, which is smaller than the theoretical upper bound of $N^{-0.24}$ given by Corollary~\ref{cor:conv_rate}. Similarly, for $\gamma = N^{0.5}$, the approximation error converges at rates of $N^{-0.85}$ for $f_1(x)$ and $N^{-1.35}$ for the LeNet5 network, both smaller than the theoretical upper bound of $N^{-0.6}$. Additionally, a rate of $N^{-1.39}$ is achieved with $\gamma = 50$ and $f_1(x)$, with a theoretical upper bound of $N^{-1.2}$.
\section{Proof Sketch}\label{sec:proofs}
In this section, we provide a proof sketch for the main theorems. The formal proof can be found in Appendix. 
To prove Theorem~\ref{th:impos}, we show that for the  function $f(x) = x$, and the case $\gamma = \mu N$ for some $\mu \in (0, 1)$, there is no encoder and decoder functions that $ \lim_{N \to \infty} \mathcal{R}(\hat{f}) = 0$.
In this case, $f(\encs(x)) = \encs(x)$, which is a spline function defined on the set $\{\alpha_k\}_{k=1}^K$. First, note that if $\lim_{N \to \infty} \declamb \neq 0$, the proof becomes trivial, as the decoder function converges to a second-degree polynomial as $N$ increases, while the encoder function remains a piecewise degree-3 polynomial.
 Therefore, we assume that $\lim_{N \to \infty} \declamb = 0$.
Let $\alpha_{\min} = \max(0, \alpha - \frac{\gamma}{2})$ and $\alpha_{\max} = \min(1, \alpha + \frac{\gamma}{2})$.
Assume that $\encs(\alpha_{\lfloor\frac{K}{2}\rfloor}) - x_{\lfloor\frac{K}{2}\rfloor} = \delta$. Define
$y_a = \encs(\alpha_{\lfloor\frac{K}{2}\rfloor}) + \Delta + 2|\delta|$. Consider the set $\mathcal{V} = \{i : \beta_i \in [\alpha_{\text{min}}, \alpha_{\text{max}}]\}$. It can be verified that $|\mathcal{V}| \leq \mu$. 

Let $P(\cdot)$ be a degree-7 polynomial that satisfies $P^{(j)}(\alpha_{\text{min}}) = \encs^{(j)}(\alpha_{\text{min}})$, $P^{(j)}(\alpha_{\text{max}}) = \encs^{(j)}(\alpha_{\text{max}})$, for $j \leqslant 2$, and  $P(\alpha_{\lfloor\frac{K}{2}\rfloor}) = y_a$.
With seven degrees of freedom in $P(\cdot)$ and these seven constraints, the adversary can configure $P(\cdot)$ to satisfy all conditions. The adversary then changes the computed value of worker node $i \in \mathcal{V}$ from $f(\encs(\beta_i))$ to $P(\beta_i)$. As a result, the master node receives 
$\overline{\mathbf{y}}_i = f(\encs(\beta_i)$, for $\beta_i \notin [\alpha_{\text{min}}, \alpha_{\text{max}}]$, and 
$P(\beta_i)$, otherwise,
from the worker nodes. 

Consider the following function 
\begin{align}
\widetilde{f}(x) = 
\begin{cases} 
f(\encs(x)) & \text{if } x \notin [\alpha_{\text{min}}, \alpha_{\text{max}}], \\
P(x) & \text{if } x \in [\alpha_{\text{min}}, \alpha_{\text{max}}].
\end{cases}
\end{align}
It can be verified that $\widetilde{f}(\cdot)$ has bounded first and second derivatives, and thus $\widetilde{f}(\cdot) \in \hil{2}$. Using the error bound for smoothing splines in Sobolev space functions (see Lemma~\ref{lem:sm_bound}), we have:  
$$
|\decs(\alpha_{\lfloor\frac{K}{2}\rfloor}) - \widetilde{f}(\alpha_{\lfloor\frac{K}{2}\rfloor})| \leqslant \sup_z |\decs(z) - \widetilde{f}(z)| \leqslant \mathcal{O}(\declamb^\frac{3}{4} + N^{-3}).
$$
Since $\declamb > C_{\lambda}N^{-4}$ and $\lim_{N \to \infty} \declamb = 0$, for every $|\delta| > \epsilon > 0$, there exists $n_0 \in \mathbb{N}$ such that for $N > n_0$,  
$
|\decs(\alpha_{\lfloor\frac{K}{2}\rfloor}) - \widetilde{f}(\alpha_{\lfloor\frac{K}{2}\rfloor})| \leqslant \epsilon < |\delta|.
$ Thus, we have:
$$
    \decs(\alpha_{\lfloor\frac{K}{2}\rfloor}) - x_{\lfloor\frac{K}{2}\rfloor} = \decs(\alpha_{\lfloor\frac{K}{2}\rfloor}) -  \encs(\alpha_{\lfloor\frac{K}{2}\rfloor}) + \delta = \decs(\alpha_{\lfloor\frac{K}{2}\rfloor}) -  \widetilde{f}(\alpha_{\lfloor\frac{K}{2}\rfloor}) + \Delta + 2|\delta| + \delta 
     \geqslant \Delta + |\delta| + \delta \geqslant \Delta.
$$
% based on Lemma~\ref{lem:sm_bound}, the decoder $\decs(\cdot)$ converges to $\widetilde{f}(x)$. More precisely, for any $0<\epsilon < \Delta$, there exist some $n_0$, such that for any $N > n_0$, we have $|\decs(\alpha_{\frac{K}{2}}) - x_{\lfloor\frac{K}{2}\rfloor}| < \epsilon$.
% Thus, we have:
% \begin{align}
% \left|\decs\left(\alpha_{\lfloor\frac{K}{2}\rfloor}\right) - x_{\lfloor\frac{K}{2}\rfloor}\right| = \left|P\left(\alpha_{\lfloor\frac{K}{2}\rfloor}\right) - x_{\lfloor\frac{K}{2}\rfloor}\right| > \Delta.
% \end{align}
%It is important to note that if $\frac{\gamma}{N} \to 0$ as $N \to \infty$, the above proof no longer holds. This is because the interval $[\alpha_{\text{min}}, \alpha_{\text{max}}]$ converges to zero as $N \to \infty$, failing to satisfy the conditions in \eqref{eq:impos_cond}.

To prove Theorem~\ref{th:dec_bound}, note that since $|f'| \leqslant \nu$, for some $\nu \in \mathbb{R}_+$,  function $f$ is $\nu$-Lipschitz. Thus:
\begin{align}\label{eq:enc_bound_}
    \lenc &= \frac{2}{K} \sum^K_{k=1} \left(f(\encs(\alpha_k))- f(x_k)\right)^2 \leqslant \frac{2\nu^2}{K} \sum^K_{k=1}  (\encs(\alpha_k)- x_k)^2.
\end{align}
Let $\decs^\circ(\cdot)$ represent the decoder function if all server nodes were honest,  i.e., $\bar{y}_n= f(\encs(\beta_n))$, for all $n \in [N]$. Thus, from \eqref{eq:decoder_opt_adv}, we have:
\begin{align}\label{eq:decoder_opt_}
    \decs^\circ=\underset{u \in \hilR{2}{\Omega}}{\operatorname{argmin}} \frac{1}{N} \sum^N_{n=1}\left(u\left(\beta_n\right)- f(\encs(\beta_n))\right)^2+\declamb \int_{\Omega} u''(t)^2\,dt.
\end{align}
Then, starting from $\ldec$ in \eqref{eq:decompose}, we have:
\begin{alignat}{2}\label{eq:proof_decompose}
   \ldec(\hat{f}) &\lec{}{=} &&\mathop{\sup}_{\mathcal{A}_\gamma}  \frac{2}{K} \sum^K_{k=1} \left|\decs \left(\alpha_k\right) - f \left(\encs\left(\alpha_k\right)\right)\right|^2 \nonumber \\
    &\lec{}{\leqslant} 2&&\mathop{\sup}_{\mathcal{A}_\gamma} \mathop{\sup}_{z}  \left|\decs \left(z\right) - f \left(\encs\left(z\right)\right)\right|^2  \nonumber \\
    &\lec{}{=} 2&&\mathop{\sup}_{\mathcal{A}_\gamma} \mathop{\sup}_{z} \left|\decs\left(z\right) - \decs^\circ\left(z\right) + \decs^\circ\left(z\right) - f \left(\encs\left(z\right)\right)\right|^2 
    \nonumber \\
    &\lec{(a)}{\leqslant} &&\underbrace{4\mathop{\sup}_{\mathcal{A}_\gamma} \mathop{\sup}_{z} |\decs\left(z\right) - \decs^\circ\left(z\right)|^2}_{\ldec^a(\hat{f})} + \underbrace{4\mathop{\sup}_{z} 
 |\decs^\circ\left(z\right) - f \left(\encs\left(z\right)\right)|^2}_{\ldec^g(\hat{f})},
\end{alignat}
where (a) follows by the AM-GM inequality. Note that,  $\ldec^g(\hat{f})$ in \eqref{eq:proof_decompose} represents the maximum decoder error in the non-adversarial regime. Thus, it represents a special case of \cite{moradicoded}, where there is no stragglers. By defining $h(t) := \decs^\circ(t) - f(\encs(t))$ and leveraging the results in \cite{moradicoded}, we can prove the following lemma:
\begin{lemma}\label{lem:sm_bound}
    Consider the proposed scheme with the same assumptions as in Theorem~\ref{th:dec_bound}. Then, we have $\ldec^g(\hat{f}) = 4 \norm{h}^2_{\lp{\infty}{\Omega}}\leqslant 4\norm{(f\circ\encs)^{(2)}}^2_{\lp{2}{\Omega}} \left(c_4\declamb^\frac{3}{4} + c_5N^{-3}\right)$,
    where $c_4, c_5$ are constants.
\end{lemma}
% \subsection{Upper bound for $\ldec^a(\hat{f})$}
To complete the proof of Theorem~\ref{th:dec_bound}, we also need to develop an upper bound for $\ldec^a(\hat{f})$ in \eqref{eq:proof_decompose}. To do so, in \cite{wahba1975smoothing} it is shown that solution of the optimization \eqref{eq:decoder_opt_} is a linear operator and can be expressed as follows: 
$$ \decs^\circ(x) = \frac{1}{N} \sum_{i=1}^N G_{N,\declamb}(x, \beta_i) y_i,
$$
where, $y_i := f(\encs(\beta_i))$ for $i \in [N]$. The weight function $G_{N,\declamb}(\cdot)$ has a complex dependence on the  points $\{\beta_i\}_{i=1}^N$ and the smoothing parameter $\declamb$. Consequently, obtaining an explicit expression for it, is quite challenging \cite{silverman1984spline,messer1991comparison}. However, in \cite{messer1991comparison, messer1993new, nychka1995splines} $G_{N,\declamb}(x, t)$ is approximated by a well-defined kernel function $K_{\declamb}(x, t)$ for the case where $\{\beta_i\}^N_{i=1}$ are equidistant. Thus, we have:
$$\decs(x) \approx \frac{1}{N} \sum_{i=1}^N K_{\declamb}(x, \beta_i) \overline{y}_i, \quad \decs^\circ(x) \approx \frac{1}{N} \sum_{i=1}^N K_{\declamb}(x, \beta_i) y_i.
$$
Here, the kernel function $K_{\declamb}(x, \beta)$ has a bandwidth of $\mathcal{O}(\declamb^{\frac{1}{4}})$ and its absolute value is bounded by $\tau\declamb^{\frac{1}{4}}\exp(|x-\beta|\declamb^{-\frac{1}{4}})$, for some constant $\tau > 0$ \cite{messer1993new, messer1991comparison}. 

Using the properties of the $K_{\declamb}(\cdot, \cdot)$ outlined in \cite{messer1993new, messer1991comparison}, it can be shown that if $\{\beta_i\}_{i=1}^N$ are equidistant and $N\declamb^{\frac{1}{4}} > c_0$ for some constant $c_0 > 0$, then there exist positive constants $c_3$ and $c_4$ depending on $c_0$, such that: 
$$\sup _{0 \leq x, \beta \leq 1} \left|G_{N,\declamb}(x, \beta)-K_{\declamb}(x, \beta)\right| < \frac{c_3}{N} + c_4\declamb^{-\frac{1}{4}}\exp({-\frac{1}{\sqrt{2}}\declamb^{-\frac{1}{4}}}).
$$
% The following lemma characterizes the maximum distance between the kernel function and the true weight function.
% \begin{lemma}\label{lem:ker_dist_weight_}
%     For $N$ equidistance points in $[0,1]$, if $N\declamb^{\frac{1}{4}} > c_0$ for constant $c_0 > 0$ and $0 < \declamb \leqslant 1$, then there exist positive constants $c_3$ and $c_4$ depending on $c_0$, such that $\sup _{0 \leq x, \beta \leq 1} \left|G_{N,\declamb}(x, \beta)-K_{\declamb}(x, \beta)\right| < \frac{c_3}{N} + c_4\declamb^{-\frac{1}{4}}e^{-\frac{1}{\sqrt{2}}\declamb^{-\frac{1}{4}}}$.
% \end{lemma}

Using the kernel representation of $\decs(\cdot), \decs^\circ(\cdot)$, we have:
\begin{alignat*}{2}
    |\decs\left(z\right) - \decs^\circ(z)| 
     &\lec{}{=} 
      |\frac{1}{N} \sum_{i=1}^N G_{N,\lambda}(z, \beta_i)(\overline{y}_i - y_i)  | \nonumber \\
     &\lec{}{=} 
      \left|\frac{1}{N} \sum_{i=1}^N \left(G_{N,\lambda}(z, \beta_i) - K_\lambda(z, \beta_i)\right)\overline{y}_i +\frac{1}{N} \sum_{i=1}^N K_\lambda(z, \beta_i)(\overline{y}_i -y_i) 
      +\frac{1}{N} \sum_{i=1}^N \left(K_\lambda(z, \beta_i) - G_{N, \lambda}(z, \beta_i)\right){y_i} \right| 
     \nonumber \\
     &\lec{}{\leqslant} 
     \frac{1}{N} \sum_{i=1}^N \left|y_i - \overline{y}_i\right| \left|G_{N,\lambda}(z, \beta_i) -  K_\lambda(z, \beta_i) \right|+  \frac{1}{N} \sum_{i=1}^N \left|y_i - \overline{y}_i\right| \left|K_\lambda(z, \beta_i)\right| 
      \nonumber \\
     &\lec{}{\leqslant}
     \frac{2M\gamma}{N} \sup_{z, \beta \in \{\beta_i\}^N_{i=1}} \left|G_{N,\lambda}(z, \beta) -  K_\lambda(z, \beta) \right| +  \frac{2M\gamma}{N} \sup_{z,  \beta \in \{\beta_i\}^N_{i=1}} \left|K_\lambda(z, \beta)\right|
      \nonumber \\
     &\lec{(a)}{\leqslant}
     \frac{2M\gamma}{N} \left(\frac{c_3}{N} + c_4\declamb^{-\frac{1}{4}}e^{-\frac{1}{\sqrt{2}}\declamb^{-\frac{1}{4}}} + \tau\declamb^{\frac{1}{4}}e^{|x-\beta|\declamb^{-\frac{1}{4}}}\right),
\end{alignat*}
where (a) follows from $\sup_{z, \beta} \left|K_\lambda(z, \beta)\right| \leqslant \tau\declamb^{\frac{1}{4}}\exp(|x-\beta|\declamb^{-\frac{1}{4}})$.
Using the above result with Lemma~\ref{lem:sm_bound} and \eqref{eq:enc_bound_} completes the proof of Theorem~\ref{th:dec_bound}.

% \subsection{Proof of Theorem~\ref{th:enc_bound}}
To prove Theorem~\ref{th:enc_bound}, we first prove that 
    If $|f'| \leqslant \nu$ and $|f''| \leqslant \eta$, then there exists a monotonically increasing function $\xi: [0, \infty) \to \mathbb{R}$ such that:
    $
        \norm{(f \circ \encs)''}^{2}_{\lp{2}{\Omega}} \leqslant (\eta^2 + \nu^2) \cdot \xi\left(\norm{\encs}^2_{\hiltilde{2}}\right).
    $
As a result, defining $\psi(t):=2 + \xi(t)$ and $\enclamb:=\max\{C_1 \frac{M^2\gamma^2}{N^4},
C_2\frac{M^2\gamma^2}{N^2}\declamb^{-\frac{1}{2}}( e^{({\sqrt{2}}\declamb^{-\frac{1}{4}})} + C_3),(\nu^2+\eta^2)(C_4\declamb^\frac{3}{4} + C_5N^{-3})\}$ completes the proof of the Theorem~\ref{th:enc_bound}.

\section*{Acknowledgment}
This material is based upon work supported by the National Science Foundation under Grant CIF-2348638.

%%%%%%
%% To balance the columns at the last page of the paper use this
%% command:
%%
%\enlargethispage{-1.2cm} 
%%
%% If the balancing should occur in the middle of the references, use
%% the following trigger:
%%
%\IEEEtriggeratref{7}
%%
%% which triggers a \newpage (i.e., new column) just before the given
%% reference number. Note that you need to adapt this if you modify
%% the paper.  The "triggered" command can be changed if desired:
%%
%\IEEEtriggercmd{\enlargethispage{-20cm}}
%%
%%%%%%

%%%%%%
%% References:
%% We recommend the usage of BibTeX:
%%

%%
%% where we here have assumed the existence of the files
%% definitions.bib and bibliofile.bib.
%% BibTeX documentation can be obtained at:
%% http://www.ctan.org/tex-archive/biblio/bibtex/contrib/doc/
%%%%%%
\newpage
\appendices
\section{Preliminaries: Sobolev spaces and Sobolev norms}\label{app:sobolev}
Let $\Omega$ be an open interval in $\mathbb{R}$ and $M$ be a positive integer. We denote by $\lpm{p}{\Omega}$ the class of all measurable functions $g:\mathbb{R} \to \mathbb{R}^M$ that satisfy:
\begin{align}
\int_{\Omega}|g_j(t)|^p \,dt\ < \infty, \quad \forall j\in [M],
\end{align}
where $g(\cdot) = [g_1(\cdot),\dots,g_M(\cdot)]^T$. The space $\lpm{p}{\Omega}$ can be endowed with the following norm, known as the ${L}^p$ norm:
\begin{align}
    \|g\|_{\lpm{p}{\Omega}} := \left(\sum^M_{j=1}\int_{\Omega}|g_i(t)|^p \,dt\right)^{\frac{1}{p}},
\end{align}
for $1 \leqslant p < \infty$, and
\begin{align}
    \|g\|_{\lpm{\infty}{\Omega}} := \max_{j \in [M]} \sup_{t\in \Omega} |g_j(t)|.
\end{align}
for $p=\infty$. Additionally, a function $g: \Omega \to \mathbb{R}^M$ is in $\lpmloc{p}{\Omega}$ if it lies in ${L}^p(V;\mathbb{R}^M)$ for all compact subsets $V \subseteq \Omega$. 
\begin{definition}[\textbf{Sobolev Space}] 
The \emph{Sobolev space} $\smp$ is the space of all functions $g \in \lpm{p}{\Omega}$ such that all weak derivatives of order $i$, denoted by $g^{(i)}$, belong to $\lpm{p}{\Omega}$ for $i \in [m]$. This space is endowed with the norm:
\begin{align}\label{eq:sob_norm1}
    \norm{g}_{\smp} :=  \left(\norm{g}_{\lpm{p}{\Omega}}^p  + \sum^m_{i=1} \norm{g^{(i)}}_{\lpm{p}{\Omega}}^p \right)^\frac{1}{p},
\end{align}
for $1 \leqslant p < \infty$, and
\begin{align}\label{eq:sob_norm2}
    \|g\|_{\sobm{m}{\infty}} := \max \left \{ \norm{g}_{\lpm{\infty}{\Omega}},  \max_{i\in [m]} \norm{g^{(i)}}_{\lpm{\infty}{\Omega}}\right \},
\end{align}
for $p=\infty$.
\end{definition}
Similarly, $\smploc$ is defined as the space of all functions $g \in \lpmloc{p}{\Omega}$ with all weak derivatives of order $i$ belonging to $\lpmloc{p}{\Omega}$ for $i \in [m]$. $\norm{g}_{\smploc}$ and $\norm{g}_{\soblocm{m}{\infty}}$ are defined similar to \eqref{eq:sob_norm1} and \eqref{eq:sob_norm2} respectively, using $\lploc{p}{\Omega}$ instead of $\lp{p}{\Omega}$. 

\begin{definition}\label{def:sobz}(\textbf{Sobolev Space with compact support}): Denoted by $\smpz$ collection of functions $g$ defined on interval $\Omega=(a, b)$ such that $g(a)=\mathbf{0}, g'(a) = \mathbf{0}, \dots, g^{(m-1)}(a)=\mathbf{0}$ and $\norm{g^{(m)}}_{\lp{2}{\Omega}} < \infty$. This space can be endowed with the following norm:
\begin{align}
    \norm{g}_{\smp} :=  \norm{g^{(m)}}_{\lpm{p}{\Omega}},
\end{align}
for $1 \leqslant p < \infty$, and
\begin{align}\label{eq:sob_norm2}
    \|g\|_{\sobm{m}{\infty}} :=  \norm{g^{(m)}}_{\lpm{\infty}{\Omega}},
\end{align}
for $p=\infty$.
\end{definition}

The next theorem provides an upper bound for ${L}^p$ norm of functions in the Sobolev space, which plays a crucial role in the proof of the main theorems of the paper.
\begin{theorem}[Theorem 7.34, \cite{leoni2024first}] \label{th:interp_ineq} Let $\Omega \subseteq \mathbb{R}$ be an open interval and let $g \in \soblocm{1}{1}$. Assume $1 \leqslant p,q,r \leqslant \infty$ and $r \geqslant q$. Then:
    \begin{align}
        \norm{g}_{\lpm{r}{\Omega}} \leqslant \ell^{\frac{1}{r}-\frac{1}{q}}\norm{g}_{\lpm{q}{\Omega}} + \ell^{1 - \frac{1}{p}+\frac{1}{r}}\norm{g'}_{\lpm{p}{\Omega}},
    \end{align}
    for every $0 < \ell < \mathcal{L}^1(\Omega)$.
\end{theorem}
Note that $\mathcal{L}^1(\Omega)$ in Theorem~\ref{th:interp_ineq} is the length of the interval $\Omega$. 
\begin{corollary}\label{col:interp_ineq}
    Suppose $g \in \soblocm{1}{1}$ and $\Omega \subseteq \mathbb{R}$ be an open interval. If $\frac{\norm{g}_{\lpm{2}{\Omega}}}{\norm{g'}_{\lpm{2}{\Omega}}} < \mathcal{L}^1(\Omega)$, then:
    \begin{align}
        \norm{g}_{\lpm{\infty}{\Omega}} \leqslant 2\sqrt{\norm{g}_{\lpm{2}{\Omega}}\cdot\norm{g'}_{\lpm{2}{\Omega}}}.
    \end{align}
\end{corollary}
\begin{proof}
    Substituting $p,q=2$ and $r=\infty$ in Theorem~\ref{th:interp_ineq} and optimizing over $\ell$, one can derive the optimum value of $\ell$, denoted by $\ell^*$ as
    \begin{align}
    \ell^* = \frac{\norm{g}_{\lpm{2}{\Omega}}}{\norm{g'}_{\lpm{2}{\Omega}}}.
    \end{align}
    Since $\frac{\norm{g}_{\lpm{2}{\Omega}}}{\norm{g'}_{\lpm{2}{\Omega}}} < \mathcal{L}^1(\Omega)$,  the optimum value is in the valid interval mentioned in Theorem~\ref{th:interp_ineq}.
\end{proof}

\begin{corollary}[Corollary 7.36, \cite{leoni2024first}] \label{cor:interp_ineq_2}
     Let $\Omega=(a, b)$, let $1 \leq p, q, r \leq \infty$ be such that $1+1 / r \geq$ $1 / p$ and $r \geq q$ and let $g \in \soblocm{1}{1}$ with $g^{\prime} \in \lpm{p}{\Omega}$. Let $x_0 \in[a, b]$ be such that $\left|g\left(x_0\right)\right|=\min _{[a, b]}|g|$. Then
\begin{align}
\left\|g-g\left(x_0\right)\right\|_{\lpm{r}{\Omega}} \leq 8\|g\|_{\lpm{q}{\Omega}}^\alpha\left\|g^{\prime}\right\|_{\lpm{p}{\Omega}}^{1-\alpha}
\end{align}
where $\alpha:=0$ if $r=q$ and $1-1 / p+1 / r=0$ and otherwise
$$
\alpha:=\frac{1-1 / p+1 / r}{1-1 / p+1 / q}.
$$
\end{corollary}

\begin{corollary}\label{col:interp_ineq2}
\cite[Corollary~3]{moradicoded} Theorem~\ref{th:interp_ineq} and Corollary~\ref{cor:interp_ineq_2} hold true when $f \in \sobm{2}{2}$.
\end{corollary}
{\bf Equivalent norms.} There have been various norms defined on Sobolev spaces in the literature that are equivalent to  \eqref{eq:sob_norm1} (see \cite{adams2003sobolev}, \cite[Ch. 7]{berlinet2011reproducing}, and \cite[Sec. 10.2]{wahba1990spline}). Note that two norms $\norm{\cdot}_{W_1}, \norm{\cdot}_{W_2}$ are equivalent if there exist positive constants $\eta_1, \eta_2$ such that
$
\eta_1\cdot\norm{g}_{W_2} \leqslant \norm{g}_{W_1} \leqslant \eta_2\cdot\norm{g}_{W_2}
$. The equivalent norm in which we are interested is the one introduced in \cite{kimeldorf1971some}. Let $\Omega = (a,b) \subset \mathbb{R}$. We define $\smpeq$ as the Sobolev space endowed with the following norm:
\begin{align}\label{eq:sob_norm_eq}
    \norm{g}_{\smpeq} := \left(\sum^M_{j=1} \left(g_j(a)^p +\sum^{m-1}_{i=1} \left(g_j^{(i)}(a)\right)^p \right) + \norm{g^{(m)}}_{\lpm{p}{\Omega}}^p \right)^\frac{1}{p}.
\end{align}
The following lemma derives the equivalence  constants ($\eta_1,\eta_2$) for the norms $\norm{\cdot}_{\sob{2}{2}}$ and $\norm{\cdot}_{\sobeq{2}{2}}$.
\begin{lemma}\label{lem:norm-equiv}
\cite[Lemma~1]{moradicoded}
    Let $\Omega = (a,b)$ be an arbitrary open interval in $\mathbb{R}$. Then for every $g \in \sob{2}{2}$:
    \begin{align}
        \label{lem:norm-equiv1}
        &\norm{g}^2_{\sob{2}{2}}\leqslant \left[2(b-a)\max\{1, (b-a)\} \left (2\max\{1, (b-a)\}^2 + 1 \right) +1\right] \cdot\norm{g}^2_{\sobeq{2}{2}} \\ \label{lem:norm-equiv2}
        &\norm{g}^2_{\sobeq{2}{2}}\leqslant \left(\frac{4}{(b-a)}\max\{1, (b-a)\}^2 + 1\right)\cdot \norm{g}^2_{\sob{2}{2}}.
    \end{align}
\end{lemma}
\begin{corollary}\label{cor:norm-equiv-1}
Based on Lemma~\ref{lem:norm-equiv}, for $\Omega=(0, 1)$ we have 
\begin{align}
     \frac{1}{5}\cdot\norm{g}^2_{\sobeq{2}{2}} \leqslant \norm{g}^2_{\sob{2}{2}} \leqslant 7\cdot\norm{g}^2_{\sobeq{2}{2}}.
\end{align}   
\end{corollary}

Corollary~\ref{cor:norm-equiv-1} directly follows from Lemma~\ref{lem:norm-equiv} by substituting $(a,b)=(0,1)$.
\begin{corollary}\label{cor:norm-equiv-2}
The result of Lemma~\ref{lem:norm-equiv} remains valid for multi-dimensional cases, where $g:\mathbb{R} \to \mathbb{R}^M$, for $M>1$.
\end{corollary}
Corollary~\ref{cor:norm-equiv-2} directly follows from applying Lemma~\ref{lem:norm-equiv} to each component of the function $g(\cdot)$ and using the definition of vector-valued function norm:
$$
\norm{g}^2_{\sobm{2}{2}} = \sum^M_{j=1}\norm{g^{(j)}}^2_{\sob{2}{2}}, \quad \norm{g}^2_{\sobmeq{2}{2}} = \sum^M_{j=1}\norm{g^{(j)}}^2_{\sobeq{2}{2}}.
$$

\begin{proposition}\label{prop:sob_hilb} 
(\cite[Section 7.2]{leoni2024first}, \cite[Theorem 121]{berlinet2011reproducing},\cite{wahba1990spline}) For any open interval $\Omega \subseteq \mathbb{R}$ and $m, M \in \mathbb{N}$, 
\begin{align}
    \hilm{m}{M} &:= \sobm{m}{2}, \nonumber \\ \hilmtilde{m}{M} &:= \sobmeq{m}{2},\nonumber \\ \hilmz{m} &:= \sobmz{m}{2},
\end{align}
are Reproducing Kernel Hilbert Spaces (RKHSs).
\end{proposition}
The full expression of the kernel function of $\hilR{m}{\Omega}$ and $\hiltilde{m}$ and other equivalent norms of Sobolev spaces can be found in \cite[Section 4]{berlinet2011reproducing}. For $\hiltilde{m}$ the kernel function is as follows:
\begin{align}\label{eq:sob:kernel}
    R(t, s) = \sum_{j=0}^{m-1} \frac{t^j s^j}{j!^2}+\int_\Omega \frac{(t-x)_{+}^{m-1}(s-x)_{+}^{m-1}}{(m-1)!^2}\, dx,
\end{align}
where $(\cdot)_{+}$ is positive part function.
\section{Smoothing Splines}\label{app:smoothspline}
Consider the data model $y_i = f(t_i) + \epsilon_i$ for $i=1,\dots,n$, where $t_i \in \Omega=(a,b) \subset \mathbb{R}$, $\mathbb{E}[\epsilon_i] = 0$, and $\mathbb{E}[\epsilon_i^2] \leqslant \sigma_0^2$. 
Assuming $f\in \sobeq{m}{2}$, the solution to the following optimization problem is referred to as the smoothing spline:
\begin{align}\label{eq:sm_spline_obj}
    \spline[\mathbf{y}] :=\underset{g \in \sobeq{m}{2}}{\operatorname{argmin}} \frac{1}{n} \sum^n_{i=1}\left(g\left(t_i\right)-y_i\right)^2+\lambda \int_{\Omega} \left(g^{(m)}(t)\right)^2\,dt,
\end{align}
where $\mathbf{y} = [y_1,\dots, y_n]$. Based on Proposition~\ref{prop:sob_hilb}, $\hiltilde{m}:=\sobeq{m}{2}$ with the norm $\norm{\cdot}_{\sobeq{m}{2}}$ is a RKHS for some kernel function $\phi(\cdot,\cdot)$. Therefore for any $v \in \sobeq{m}{2}$, we have:
\begin{align}\label{eq:smothspline_kernel}
    v(t) = \langle v(\cdot), \phi(\cdot, t)\rangle_{\hiltilde{m}}.
\end{align}
It can be shown that $\phi(t,s) = R^P(t, s) + \phi_0(t, s)$ where $\phi_0(t, s)$ is kernel function of $\hilz{m}$ and $R^P(t,s)$ is a null space of $\hilz{m}$ which is the space of all polynomials with degree less than $m$.

The solution of \eqref{eq:sm_spline_obj} has the following form \cite{wahba1990spline,duchon1977splines}:
\begin{align}\label{eq:spline_sol_gen}
    u^*(\cdot) = \sum^m_{i=1} d_i \zeta_i(\cdot) + \sum^n_{j=1} c_j \phi_0(\cdot, t_j),
\end{align}
where $\left\{\zeta_i(\cdot)\right\}_{i=1}^m$ are the basis functions of the space of polynomials of degree at most $m-1$. Substituting $u^*$ into \eqref{eq:sm_spline_obj} and optimizing over $\mathbf{c} = [c_1,\dots, c_n]^T$ and $\mathbf{d} = [d_1,\dots, d_m]^T$, we obtain the following result \cite{wahba1990spline}:
\begin{align}\label{eq:spline_linear}
    \spline[\mathbf{y}](\mathbf{y}) &= \mathbf{Q}\left(\mathbf{Q}^{T} \mathbf{Q}+\lambda \mathbf{\Gamma}\right)^{-1} \mathbf{Q}^{T} \mathbf{y}, \\ \label{eq:spline_coef_d}
    \mathbf{d}&=\left(\mathbf{P}^{T} \mathbf{L}^{-1} \mathbf{P}\right)^{-1} \mathbf{P}^{T} \mathbf{L}^{-1} \mathbf{y}, \\ \label{eq:spline_coef_c}
    \mathbf{c}&=\mathbf{L}^{-1}\left(\mathbf{I}-\mathbf{P}\left(\mathbf{P}^{T} \mathbf{L}^{-1} \mathbf{P}\right)^{-1} \mathbf{P}^{T} \mathbf{L}^{-1}\right) \mathbf{y},
\end{align}
where
\begin{align}
\mathbf{Q}_{n \times(n+m)}&=\left[\begin{array}{ll}
\mathbf{P}_{n \times m} & \boldsymbol{\Sigma}_{n \times n}
\end{array}\right], \nonumber \\
\mathbf{\Gamma}_{(n+m) \times(n+m)}&=\left[\begin{array}{cc}
\mathbf{0}_{m \times m} & \mathbf{0}_{m \times n} \nonumber  \\
\mathbf{0}_{n \times 1} & \boldsymbol{\Sigma}_{n \times n}
\end{array}\right], \nonumber \\
\boldsymbol{P}_{ij} &= \zeta_j(t_i),\nonumber \\
\boldsymbol{\Sigma}_{ij} &= \phi_0(t_i, t_j), \nonumber \\
\mathbf{L} &= \mathbf{\Sigma} + n\lambda \mathbf{I}.
\end{align}
Equation \eqref{eq:spline_linear} states that the smoothing spline fitted on the data points $\mathbf{y}$ is a linear operator:
\begin{align}\label{eq:ssfunction_def}
\spline[\mathbf{y}](\mathbf{z}) := \mathbf{A_{\lambda}} \mathbf{z},
\end{align}
 for $\mathbf{z} \in \mathbb{R}^n$, where $\mathbf{A_\lambda}:=\mathbf{Q}\left(\mathbf{Q}^T \mathbf{Q}+\lambda \mathbf{\Gamma}\right)^{-1} \mathbf{Q}^{T}$. 
 
To characterize the estimation error of the smoothing spline, $|f - \spline(\mathbf{y})|$, we need to define two variables quantify the minimum and maximum consecutive distance of the regression points $\{t_i\}^n_{i=1}$:
\begin{gather}\label{eq:max_dist}
\Delta_\textrm{max}:=\underset{i\in \{0\} \cup [n]}{\max} \left\{t_{i+1}-t_i\right\}, \quad \Delta_\textrm{min}:=\underset{i\in [n-1]}{\min} \left\{t_{i+1}-t_i\right\},
\end{gather}
where boundary points are defined as $(t_0, t_{n+1}):=(a,b)$. The following theorem offers an upper bound for the $j$-th derivative of the smoothing spline estimator error function in the absence of noise ($\sigma_0 = 0$).
\begin{theorem} \label{th:lit_spline_noiseless} (\cite[Theorem 4.10]{ragozin1983error})
Consider data model $y_i = f(t_i)$ with $\{t_i\}^n_{i=1}$ belong to $\Omega=[a, b]$ for $i \in [n]$. Let
\begin{align}
L=p_{2(m-1)}(\frac{\Delta_\textrm{max}}{\Delta_\textrm{min}})\cdot\frac{n \Delta_\textrm{max}}{b-a} \frac{\lambda}{2}+D(m)\cdot \left({\Delta_\textrm{max}}\right)^{2 m},
\end{align}
where $p_{d}(\cdot)$ is a degree $d$ polynomial with positive weights and $D(m)$ is a function of $m$. Then for each  $j\in \{0,1,\dots,m\}$ and any $f\in\sob{m}{2}$, there exist a function $H(m, j)$ such that:
\begin{align}\label{eq:noiseless_bound}
    \norm{\left(f - \spline(\mathbf{y})\right)^{(j)}}_{\lp{2}{\Omega}}^2 \leqslant H(m, j)\left(1+\left(\frac{L}{(b-a)^{2m}}\right)\right)^{\frac{j}{m}}\cdot L^{\frac{(m-j)}{m}}\cdot \norm{f^{(m)}}^2_{\lp{2}{\Omega}}.
\end{align}
\end{theorem}
Note that $\left(f - \spline(\mathbf{y})\right)^{(0)} := f - \spline(\mathbf{y})$. 

Substituting \eqref{eq:spline_coef_c} and \eqref{eq:spline_coef_d} into \eqref{eq:spline_sol_gen}, one can conclude that:
\begin{align}\label{eq:spline_green_f_1}
    u^*(x) = \boldsymbol{\zeta}(x)^T\mathbf{M_1}\mathbf{y} + \boldsymbol{\phi_0}(x)^T\mathbf{M_2}\mathbf{y}, 
\end{align}
where $\boldsymbol{\phi_0}(x):=[\phi_0(x,t_1)\dots\phi_0(x,t_n)]^T$, $\boldsymbol{\zeta}(x):=[\zeta_1(x)\dots\zeta_m(x)]^T$, and matrices $M_1$ and $M_2$ are defined as $\mathbf{M_1}:=\left(\mathbf{P}^{T} \mathbf{L}^{-1} \mathbf{P}\right)^{-1} \mathbf{P}^{T} \mathbf{L}^{-1}$ and $\mathbf{M_2}\,:= \,\mathbf{L}^{-1}\left(\mathbf{I}-\mathbf{P}\left(\mathbf{P}^{T} \mathbf{L}^{-1} \mathbf{P}\right)^{-1} \mathbf{P}^{T} \mathbf{L}^{-1}\right)$. Thus, there exists a weight function, denoted by $G_{n,\lambda}(x, t)$, such that: 
\begin{align}\label{eq:spline_green_f_2} u^*(x) = \frac{1}{n} \sum_{i=1}^n G_{n,\lambda}(x, t_i)y_i,
\end{align} 
where $G_{n,\lambda}(x, t)$ depends on the regression points $\{t_i\}_{i=1}^n$ and the smoothing parameter $\lambda$. 

Obtaining an explicit expression for $G_{n,\lambda}(x, t)$ is challenging due to the complexity of the inverse matrix operation and the kernel of $\hilm{2}{1}$ in general cases. However, it has been shown that the weight function $G_{n,\lambda}(x, t)$ can be approximated by a well-defined kernel function $K_\lambda(x, t)$ under certain assumptions about the regression points ${t_i}_{i=1}^n$ \cite{silverman1984spline, messer1991comparison, messer1993new, nychka1995splines}.

In an early work, \cite{silverman1984spline} proposed the following kernel function: 
\begin{align} \label{eq:kernel_old}
\kappa(u) = \frac{1}{2} \exp\left(\frac{-|u|}{\sqrt{2}}\right) \sin\left(\frac{|u|}{\sqrt{2}} + \frac{\pi}{4}\right). 
\end{align}

\cite{silverman1984spline} demonstrated that for $t$ sufficiently far from the boundary, the following holds: \begin{align}\label{eq:eq_ker_0} \lambda^{\frac{1}{4}} q(t)^{-\frac{1}{4}} G_{n,\lambda}\left(t + \lambda^{\frac{1}{4}} q(t)^{-\frac{1}{4}} x, t\right) \xrightarrow{n \to \infty} \frac{\kappa(x)}{q(t)}, \end{align} where $q(t)$ represents the asymptotic probability density function of the regression points.

According to \eqref{eq:eq_ker_0}, as $n$ increases, the smoothing spline behaves like a kernel smoother with a bandwidth of $\lambda^{\frac{1}{4}} q(t)^{-\frac{1}{4}}$. However, this approximation was later found to be insufficient for establishing certain asymptotic properties \cite{messer1991comparison}.

Consider the continuous formulation of the smoothing spline objective over the domain $\Omega = [0, 1]$: \begin{align}\label{eq:green} g^*(\cdot) := \underset{g \in \sobeq{m}{2}}{\operatorname{argmin}} \int_0^1 \left(f(t) - g(t)\right)^2,dt + \lambda \int_0^1 \left(g''(t)\right)^2,dt. 
\end{align} 
It can be shown that the solution to \eqref{eq:green} is characterized by a unique Green's function $G_\lambda(x, t)$, satisfying: \begin{align}\label{eq:green_2} h(x) := \int_0^1 G_\lambda(x, t) f(t),dt.
\end{align}

\cite{messer1993new} proposed an alternative approximated kernel function for equidistant regression points, i.e. $t_i = \frac{i}{n}$, which addresses the asymptotic issues of \eqref{eq:eq_ker_0}:
\begin{align} \label{eq:kernel_new}
    K_\lambda(x, t)= & \left(2 \sqrt{2}\lambda^{\frac{1}{4}}\right)^{-1} e^{-\frac{|x-t|}{\sqrt{2}\lambda^{\frac{1}{4}}}}\left(\sin \frac{|x-t|}{\sqrt{2}\lambda^{\frac{1}{4}}}+\cos \left(\frac{x-t}{\sqrt{2}\lambda^{\frac{1}{4}}}\right)\right) \nonumber \\ & +\left(2 \sqrt{2}\lambda^{\frac{1}{4}}\right)^{-1}\left(\Phi\left(\frac{x+t}{\sqrt{2}\lambda^{\frac{1}{4}}}, \frac{x-t}{\sqrt{2}\lambda^{\frac{1}{4}}}\right)+\Phi\left(\frac{1-x}{\sqrt{2}\lambda^{\frac{1}{4}}}+\frac{1-t}{\sqrt{2}\lambda^{\frac{1}{4}}}, \frac{1-x}{\sqrt{2}\lambda^{\frac{1}{4}}}-\frac{1-t}{\sqrt{2}\lambda^{\frac{1}{4}}}\right)\right),
\end{align}
where, $\Phi(u, v):= e^{-u}\left(\cos(u) - \sin(u) + 2\cos(v)\right)$. Note that the absolute values of the kernels in \eqref{eq:kernel_old} and \eqref{eq:kernel_new} decay exponentially to zero as the distance $|t - z|$ increases. The bandwidth of the kernels are $\mathcal{O}(\lambda^{\frac{1}{4}})$.
\begin{lemma}\label{lem:kernel_upperbound}
    For the the kernel function proposed in \eqref{eq:kernel_new}, there exist a constant $\tau$ such that:
    \begin{align}
        \sup _{0 \leq x, t \leq 1} \left|K_\lambda(x, t)\right| \leqslant \tau\lambda^{-\frac{1}{4}}.
    \end{align} 
\end{lemma}
\begin{proof}
    From \eqref{eq:kernel_new}, and using the bounds $|\sin(x)| \leqslant 1$ and $|\cos(x)| \leqslant 1$, we obtain: \begin{align} K_\lambda(x, t) &\leqslant \left(\sqrt{2}\lambda^{\frac{1}{4}}\right)^{-1} e^{-\frac{|x-t|}{\sqrt{2}\lambda^{\frac{1}{4}}}} + 4\left(\sqrt{2}\lambda^{\frac{1}{4}}\right)^{-1} \left(e^{-\frac{(x+t)}{\sqrt{2}\lambda^{\frac{1}{4}}}} + e^{-\frac{(2-x-t)}{\sqrt{2}\lambda^{\frac{1}{4}}}}\right) \nonumber \ &\leqslant \left(\sqrt{2}\lambda^{\frac{1}{4}}\right)^{-1} + 8\left(\sqrt{2}\lambda^{\frac{1}{4}}\right)^{-1} \nonumber \ &\leqslant \frac{9}{\sqrt{2}}\lambda^{-\frac{1}{4}}. \end{align}
\end{proof}

The following lemma that characterize the relationship between their proposed kernel $K_\lambda(x, t)$ and continuous Green function $G_\lambda(x, t)$:
\begin{lemma} \cite[Theorem~4.1]{messer1993new}\label{lem:ker_dist_green}
For $0 < \lambda \leqslant 1$, there exists constant $C$ such that:
\begin{align}
    \sup _{0 \leq x, t \leq 1}\left|K_\lambda(x, t) - G_\lambda(x, t)\right| \leq C\frac{e^{-\left( \frac{1}{\sqrt{2}} \lambda^{-\frac{1}{4}}\right)}}{\lambda^{\frac{1}{4}}}.
\end{align}
\end{lemma}
According to Lemma~\ref{lem:ker_dist_green}, if $\lambda \to 0$ as $n \to \infty$, the proposed kernel converges exponentially to the true Green's function.
Moreover, with an appropriate choice of $\lambda$, in the case of uniform regression points, it can be ensured that the weight function $G_{n,\lambda}(x, t)$ defined in \eqref{eq:spline_green_f_2} converges to the Green's function $G_\lambda(x, t)$:
\begin{lemma} \cite[Theorem~2.1]{nychka1995splines} \label{lem:ker_dist_green_2} For uniformly distributed regression points in $[0,1]$, if $n\lambda^{\frac{1}{4}} > C_0$ for some constant $C_0 > 0$, then there exist positive constants $\mu_1$ and $\mu_2$ such that:
\begin{align}
    \sup _{0 \leq x, t \leq 1} \left|G_{n,\lambda}(x, t)-\mathrm{G}_\lambda(x, t)\right|<\frac{\mu_1n^{-1}}{1-\mu_1n^{-1}\lambda^{-\frac{1}{4}}} \mathrm{e}^{-{\mu_2\left|t-x\right|}{\lambda^{-\frac{1}{4}}}}.
\end{align}
\end{lemma}
We now present our main lemma, which characterizes the relationship between the weight function $G_{n,\lambda}(x, t)$ and the approximated kernel $K_\lambda(x, t)$:
\begin{lemma}\label{lem:ker_dist_weight}
    For uniformly distributed regression points in $[0,1]$, if $n\lambda^{\frac{1}{4}} > C_0$ for some constant $C_0 > 0$ and $0 < \lambda \leqslant 1$, then there exist positive constants $\mu_3$ and $\mu_4$ depending on $C_0$, such that:
    \begin{align}
        \sup _{0 \leq x, t \leq 1} \left|G_{n,\lambda}(x, t)-K_\lambda(x, t)\right| < \frac{\mu_3}{n}e^{\left(-\lambda^{-\frac{1}{4}}|x-t|\right)} + \mu_4\lambda^{-\frac{1}{4}}e^{\left(-\frac{1}{\sqrt{2}}\lambda^{-\frac{1}{4}}\right)}.
    \end{align}
\end{lemma}
\begin{proof}
    Lemma~\ref{lem:ker_dist_weight} can directly derive from Lemma~\ref{lem:ker_dist_green} and Lemma~\ref{lem:ker_dist_green_2} and using the fact that:
    \begin{align}
        \sup _{0 \leq x, t \leq 1} \left|G_{n,\lambda}(x, t)-K_\lambda(x, t)\right| &\leqslant  \sup _{0 \leq x, t \leq 1} \left|G_{n,\lambda}(x, t)- G_\lambda(x, t)| + |G_\lambda(x, t) +K_\lambda(x, t)\right| \nonumber \\
        &\leqslant 
        \sup _{0 \leq x, t \leq 1} |G_{n,\lambda}(x, t)- G_\lambda(x, t)| + \sup _{0 \leq x, t \leq 1} |G_\lambda(x, t) +K_\lambda(x, t)|
    \end{align}
\end{proof}

\section{Proof of Theorem~\ref{th:dec_bound}}\label{sec:app_th_dec_proof}
Recall from \eqref{eq:decompose} and \eqref{eq:proof_decompose} that $\mathcal{R}(\hat{f}) \leqslant \lenc(\hat{f}) + \ldec^a(\hat{f}) + \ldec^g(\hat{f})$, where 
\begin{align}
    \label{eq:adv_bound_proof}
    \ldec^a(\hat{f}) &= 4\mathop{\sup}_{\mathcal{A}_\gamma}   \mathop{\sup}_{z} \left|\decs\left(z\right) - \decs^\circ(z)\right|^2, \\ \label{eq:dec_bound_proof}   \ldec^g(\hat{f}) &= 4\mathop{\sup}_{z} 
 \left|\decs^\circ\left(z\right) - f(\encs(z))\right|^2,  \\   
 \label{eq:enc_bound_proof} \lenc(
 \hat{f}) &= \frac{2}{K} \sum^K_{k=1} \left|f(\encs(\alpha_k)) - f(x_k)\right|^2.
\end{align}
Let us begin with $\lenc$. Since $|f'|\leqslant \nu$, $f(\cdot)$ is $\nu$-Lipschitz property. Thus, we have:
\begin{align}\label{eq:lenc_final}
     \lenc &= \frac{2}{K} \sum^K_{k=1} \left(f(\encs(\alpha_k))- f(x_k)\right)^2 \nonumber \\ 
     &= \frac{2}{K} \sum^K_{k=1} \left(|f(\encs(\alpha_k))- f(x_k)|\right)^2 \nonumber \\ 
     &\leqslant \frac{2}{K} \sum^K_{k=1} \left(\nu\cdot|\encs(\alpha_k)- x_k|\right)^2 \nonumber \\
     &= \frac{2\nu^2}{K} \sum^K_{k=1}  (\encs(\alpha_k)- x_k)^2.
\end{align}
Next, we continue to $\ldec^g(\hat{f})$. As noted earlier, $\ldec^g(\hat{f})$ quantifies the how generalized the decoder function is. To apply the results from Theorem~\ref{th:lit_spline_noiseless}, it is necessary to ensure that $f\circ\encs \in \sob{2}{2}$.
\begin{lemma} \cite[Lemma~3]{moradicoded}\label{lem:fou}
Let $f:\mathbb{R} \to \mathbb{R}$ be a $\nu$-Lipschitz continuous function with $\norm{f''}_{\lp{\infty}{\Omega}} \leqslant \eta$ and $\Omega \subset \mathbb{R}$ be an open interval. If $\encs \in \sob{2}{2}$ then $f\circ \encs \in \sob{2}{2}$.
\end{lemma}
Define the function $h(t) := \decs^\circ(t) - f(\encs(t))$. By Lemma~\ref{lem:fou}, and under the assumption that both $\decs^\circ$ and $f \circ \encs$ belong to the Sobolev space $\sob{2}{2}$, it follows that $h \in \sob{2}{2}$. The subsequent lemmas derive upper bounds for $\norm{h}_{\lp{\infty}{\Omega}}$ and $\norm{h'}_{\lp{\infty}{\Omega}}$ by utilizing the fundamental properties of Sobolev spaces.
\begin{lemma}\label{lem:norm_infty_bound_1}
    If $\Omega=(a, b)$, then:
\begin{align}
    \|h\|^2_{\lp{2}{\Omega}} \leqslant\left\|h^{\prime}\right\|^2_{\lp{2}{\Omega}} \cdot \left(x_{0}^{2} - x_0(a+b)+\frac{a^{2} + b^2}{2}\right)
\end{align}
\end{lemma}
\begin{proof}
Assume $\exists\, x_{0} \in \Omega: h\left(x_{0}\right)=0$. Therefore, $|h(x)|=\left|\int_{x_{0}}^{x} h^{\prime}(x)\,dx\right|$ for $x \in\left[x_{0}, b\right)$. Thus
\begin{align}
    |h(x)| = \left|\int_{x_{0}}^{x} h^{\prime}(x)\,dx\right| \lec{(a)}{\leqslant} \int_{x_{0}}^{x}\left|h^{\prime}(x)\right|\,dx \lec{(b)}{\leqslant} \left(\int_{x_{0}}^{x} 1^{2}\,dx\right)^{\frac{1}{2}}\left(\int_{x_{0}}^{x} |h^{\prime}(x)|^2\,dx\right)^{\frac{1}{2}},
\end{align}
where (a) and (b) are followed by the triangle and Cauchy-Schwartz inequalities respectively. Integrating the square of both sides over the interval $[x_0, b)$ yields:
\begin{align}\label{eq:norm_infty_bound_1}
 \int_{x_{0}}^{b}|h(x)|^{2}\,dx \lec{}{\leqslant} \int_{x_{0}}^{b}\left(z-x_{0}\right) \cdot \left( \int_{x_{0}}^{b}\left|h^{\prime}(x)\right|^{2} \,dx \right) \,dz \lec{(a)}{\leqslant} \int_{x_{0}}^{b}\left(z-x_{0}\right)\,dz \cdot \int_{a}^{b}\left|h^{\prime}(x)\right|^{2}\,dx,
\end{align}
where (a) follows by $x_0 \geqslant a$. On the other side, we have the following for every $x \in\left(a, x_{0}\right]$:
\begin{align}
& |h(x)|=\left|\int_{x}^{x_{0}} h^{\prime}(x)\,dx\right| \leqslant \int_{x}^{x_{0}}\left|h^{\prime}(x)\right|\,dx \leqslant\left(\int_{x}^{x_{0}} 1^{2}\,dx\right)^{\frac{1}{2}} \cdot\left(\int_{x}^{x_{0}}\left|h^{\prime}(x)\right|^{2}\,dx\right)^{\frac{1}{2}}.
\end{align}
Therefore, we have a similar inequality:
\begin{align}\label{eq:norm_infty_bound_2}
& \qquad \int_{a}^{x_{0}}|h(x)|^{2}\,dx \leqslant \int_{a}^{x_{0}}\left(x_{0}-x\right)\,dx \cdot \int_{a}^{b}|h'(x)|^{2}\,dx.
\end{align}
Using \eqref{eq:norm_infty_bound_1} and \eqref{eq:norm_infty_bound_2} completes the proof:
\begin{align}
\|h\|_{L^{2}(\Omega)}^{2}&=\int_{a}^{x_{0}}|h(x)|^{2}\,dx+\int_{x_{0}}^{b}|h(x)|^{2} \,dx \nonumber \\ 
&\leqslant\left\|h^{\prime}\right\|_{\lp{2}{\Omega}}^{2} \cdot \left(\int_{a}^{x_{0}}\left(x_{0}-x\right)\,dx+\int_{x_{0}}^{b}\left(x-x_{0}\right)\,dx\right) \nonumber \\
&\leqslant\left\|h^{\prime}\right\|_{\lp{2}{\Omega}}^{2} \cdot\left(x_{0}\left(x_{0}-a\right)-\left(\frac{x_{0}^{2}-a^2}{2}\right)+\left(\frac{b^{2}-x_{0}^{2}}{2}\right)-x_{0}\left(b-x_{0}\right)\right) \\
&= \left\|h^{\prime}\right\|_{\lp{2}{\Omega}}^{2} \cdot\left(x_{0}^{2} - x_0(a+b)+\frac{a^{2} + b^2}{2}\right).
\end{align}
Thus, if $x_0$ exists, the proof is complete. In the next step, we prove the existence of  $x_0\in\Omega$ such that $h\left(x_{0}\right)=0$. Recall that $h(t) = \decs^\circ(t) - f(\encs(t))$ and $\decs^\circ(\cdot)$ is the solution of \eqref{eq:decoder_opt_}. Assume there is no such $x_0$. Since $h \in \sob{2}{2}$, then $h(\cdot)$ is continuous. Therefore, if there exist $t_1,t_2 \in \Omega$ such that $h(t_1) < 0$ and $h(t_2) > 0$, then the intermediate value theorem states that there exists $x_0 \in (t_1, t_2)$ such that $h(x_0)=0$. Thus, $h(t) >0$ or $h(t)<0$ for all $t\in\Omega$. Without loss of generality, assume the first case where $h(t) > 0$ for all $t\in\Omega$. It means that $\decs^\circ(t) > f(\encs(t))$ for all $t\in\Omega$. Let us define $$\beta^* := \underset{\beta \in [N]}{\operatorname{argmin}}\, \decs^\circ(\beta) - f(\encs(\beta)).$$ Let $\decsbar^\circ(t) := \decs^\circ(t) - \decs^\circ(\beta^*)$. Note that $\int_\Omega \left(\decsbar^{0''}(t)\right)^2\,dt = \int_\Omega \left(\decs^{0''}(t)\right)^2\,dt$. Therefore, 
\begin{align}
    \sum_{v \in [N]}\left[\decs^\circ\left(\beta_v\right)-f\left(\encs\left(\beta_v\right)\right)\right]^2 &= \sum_{v \in [N]}\left[\decsbar^\circ \left(\beta_v\right) + \decs^\circ(\beta^*)-f\left(\encs\left(\beta_v\right)\right)\right]^2 \nonumber \\
    &= \sum_{v \in [N]}\left[\decsbar^\circ\left(\beta_v\right)-f\left(\encs\left(\beta_v\right)\right)\right]^2 + N\cdot\decs^\circ(\beta^*)^2 \nonumber \\
    &\quad\quad\quad + 2\decs^\circ(\beta^*) \sum_{v\in [N]}[\decsbar^\circ(\beta_v) - f(\encs(\beta_v))] \nonumber \\
    &\lec{(a)}{\geqslant} \sum_{v \in [N]} \left[\decsbar^\circ\left(\beta_v\right)-f\left(\encs\left(\beta_v\right)\right)\right]^2,
\end{align}
where (a) follows from $\decs^\circ(\beta^*) > 0$ and $\decsbar^\circ(\beta_v) > f(\encs(\beta_v))$ for all $v \in [N]$. This leads to a contradiction since it implies that $\decs^\circ$ is not the solution of the \eqref{eq:decoder_opt_}. Therefore, our initial assumption must be wrong. Thus, there exists $x_0\in \Omega$ such that $h(x_0) = 0$.
\end{proof}
\begin{lemma}\label{lem:norm_infty_bound_2}
Let $\Omega=(0,1)$. For $h(t) = \decs^\circ(t) - f(\encs(t))$ we have:
\begin{align}\label{eq:norm_infty_bound_2_1}
  \norm{h}_{\lp{\infty}{\Omega}} \leqslant 2\norm{h}^\frac{1}{2}_{\lp{2}{\Omega}}\cdot\norm{h'}^\frac{1}{2}_{\lp{2}{\Omega}} < \infty,
\end{align}
and
\begin{align}\label{eq:norm_infty_bound_2_2}
  \norm{h'}_{\lp{\infty}{\Omega}} \leqslant \norm{h'}_{\lp{2}{\Omega}} + \norm{h''}_{\lp{2}{\Omega}} < \infty.
\end{align}
\end{lemma}
\begin{proof}
    Using Lemma~\ref{lem:norm_infty_bound_1} for $a=0, b=1$, one can conclude
    \begin{align}
        \frac{\norm{h}_{\lp{2}{\Omega}}}{\norm{h'}_{\lp{2}{\Omega}}} \leqslant \frac{1}{\sqrt{2}}.
    \end{align}
    Since $h \in \sob{2}{2}$ we can apply Corollary~\ref{col:interp_ineq} and Theorem~\ref{th:interp_ineq} with $r=\infty$ and $p,q=2$ to complete the proof of \eqref{eq:norm_infty_bound_2_1}. Furthermore, using Theorem~\ref{th:interp_ineq} with $r=\infty$ and $p,q=2$ and $\ell = 1$ completes the proof of \eqref{eq:norm_infty_bound_2_2}.
\end{proof}
Using Lemma~\ref{lem:norm_infty_bound_2} and beginning with \eqref{eq:dec_bound_proof}, an upper bound for $\ldec^g(\hat{f})$ can be obtained:
\begin{align}\label{eq:ldec_step3}
    \ldec^g(\hat{f}) &= 4\mathop{\sup}_{z\in \Omega} 
 \left|\decs^\circ\left(z\right) - f(\encs(z))\right|^2
 \nonumber \\  &\lec{}{=}  4 \norm{h}^2_{\lp{\infty}{\Omega}}
  \nonumber \\  &\lec{(a)}{\leq}  8 \norm{h}_{\lp{2}{\Omega}}\cdot \norm{h'}_{\lp{2}{\Omega}},
\end{align}
where (a) follows from Lemma~\ref{lem:norm_infty_bound_2}. Applying Theorem~\ref{th:lit_spline_noiseless} with $\Omega=(0, 1), m=2$, we have
\begin{align}\label{eq:h_1}
    \norm{h}^2_{\lp{2}{\Omega}} \lec{}{\leqslant}H_0\norm{(f\circ\encs)^{(2)}}^2_{\lp{2}{\Omega}}\cdot L,
\end{align}
and
\begin{align}\label{eq:hp_1}
\norm{h'}^2_{\lp{2}{\Omega}} &\lec{}{\leqslant}H_1\norm{(f\circ\encs)^{(2)}}^2_{\lp{2}{\Omega}}\cdot L^{\frac{1}{2}}(1 + \frac{L}{16})^{\frac{1}{2}},
\end{align}
where $L=p_2\left(\frac{\Delta_\textrm{max}}{\Delta_\textrm{min}}\right)\cdot\frac{N \Delta_\textrm{max}}{4} \declamb + D(2)\cdot \left({\Delta_\textrm{max}}\right)^{4}$ and $H_0, H_1:= H(2,0), H(2,1)$ as defined in  Theorem~\ref{th:lit_spline_noiseless}. Since, $\beta_i = \frac{i}{N}$, $\Delta_\textrm{max} = \Delta_\textrm{min}=\frac{1}{N}$. Thus, we have:
\begin{align}\label{eq:ldec_final}
    \ldec^g(\hat{f}) &\lec{}{\leqslant} \sqrt{H_0H_1} \cdot \norm{(f\circ\encs)^{(2)}}^2_{\lp{2}{\Omega}} L^\frac{3}{4} \left(1+\frac{L}{16}\right)^\frac{1}{4}
     \nonumber \\
     &\lec{}{\leqslant} \sqrt{H_0H_1} \cdot \norm{(f\circ\encs)^{(2)}}^2_{\lp{2}{\Omega}} \left(c_1\declamb + c_2N^{-4}\right)^\frac{3}{4} \left(1+c_1\declamb\ + c_2N^{-4} \right)^\frac{1}{4}
     \nonumber \\
     &\lec{(a)}{\leqslant} \sqrt{H_0H_1} \cdot \norm{(f\circ\encs)^{(2)}}^2_{\lp{2}{\Omega}} \left(c_1\declamb + c_2N^{-4}\right)^\frac{3}{4} \left(1+c_1+c_2\right)^\frac{1}{4}
     \nonumber \\
     &\lec{(b)}{\leqslant} \sqrt{H_0H_1} \cdot \norm{(f\circ\encs)^{(2)}}^2_{\lp{2}{\Omega}} \left(\left(c_3\declamb\right)^\frac{3}{4} + \left( c_4N^{-4}\right)^\frac{3}{4}\right),
\end{align}
where \( c_1 := \frac{p_2(1)}{4} \), \( c_2 := D(2) \), \( c_3 := c_1(1 + c_1 + c_2)^{\frac{1}{3}} \), and \( c_4 := c_2(1 + c_1 + c_2)^{\frac{1}{3}} \). Step (a) follows from the conditions \( N > 1 \) and \( \declamb \leqslant 1 \), while (b) is derived from the inequality \( (x + y)^r \leqslant x^r + y^r \), which holds for \( x, y > 0 \) and \( 0 < r < 1 \).

Finally, for $\ldec^a(\hat{f})$, let us define $y_i:=f(\encs(\beta_i))$. Leveraging the weight function representation of smoothing splines we have:
\begin{alignat}{2}
    \mathop{\sup}_{\mathcal{A}_\gamma}   \mathop{\sup}_{z} \left|\decs\left(z\right) - \decs^\circ(z)\right| 
     &\lec{}{=} \mathop{\sup}_{\mathcal{A}_\gamma}   \mathop{\sup}_{z}
      &&\left|\frac{1}{N} \sum_{i=1}^N G_{N,\lambda}(z, \beta_i)\overline{y}_i - \frac{1}{N} \sum_{i=1}^N G_{N,\lambda}(z, \beta_i)y_i  \right| \nonumber \\
     &\lec{(a)}{=} \mathop{\sup}_{\mathcal{A}_\gamma}   \mathop{\sup}_{z}
      &&|\frac{1}{N} \sum_{i=1}^N G_{N,\lambda}(z, \beta_i)\overline{y}_i - \frac{1}{N} \sum_{i=1}^N K_\lambda(z, \beta_i)\overline{y}_i \nonumber \\ & &&+ \frac{1}{N} \sum_{i=1}^N K_\lambda(z, \beta_i)\overline{y}_i - \frac{1}{N} \sum_{i=1}^N K_{\lambda}(z, \beta_i)y_i \nonumber \\ & &&+ \frac{1}{N} \sum_{i=1}^N K_\lambda(z, \beta_i){y_i} - \frac{1}{N} \sum_{i=1}^N G_{N, \lambda}(z, \beta_i){y_i}  | 
     \nonumber \\
     &\lec{}{\leqslant} \mathop{\sup}_{\mathcal{A}_\gamma}   \mathop{\sup}_{z}
      && \frac{1}{N} \sum_{i=1}^N \left|y_i - \overline{y}_i\right| \cdot \left|G_{N,\lambda}(z, \beta_i) -  K_\lambda(z, \beta_i) \right| \nonumber \\ & &&+  \frac{1}{N} \sum_{i=1}^N \left|y_i - \overline{y}_i\right|\cdot \left|K_\lambda(z, \beta_i)\right|,
\end{alignat}
where (a) follows from adding and subtracting two terms: $\frac{1}{N} \sum_{i=1}^N K_\lambda(z, \beta_i)\overline{y}_i, \frac{1}{N} \sum_{i=1}^N K_\lambda(z, \beta_i)y_i$. 
By defining $d_z := \min_{i \in [N]} \left|z - \beta_i\right|$, we have:
\begin{alignat}{2}
\mathop{\sup}_{\mathcal{A}_\gamma}   \mathop{\sup}_{z} \left|\decs\left(z\right) - \decs^\circ(z)\right| 
 &\lec{(a)}{\leqslant}  &&\mathop{\sup}_{z}
       \frac{2M\gamma}{N} \left(\frac{\mu_3}{N}e^{\left(-d_z\lambda^{-\frac{1}{4}}\right)} + \mu_4\lambda^{-\frac{1}{4}}e^{\left(-\frac{1}{\sqrt{2}}\lambda^{-\frac{1}{4}}\right)}\right) +  \frac{2M\gamma}{N} \frac{9}{\sqrt{2}}\lambda^{-\frac{1}{4}}
       \nonumber \\
     & \lec{}{\leqslant} 
       &&\mu_3 \frac{2M\gamma}{N^2} + \mu_4\frac{2M\gamma}{N}\lambda^{-\frac{1}{4}}e^{\left(-\frac{1}{\sqrt{2}}\lambda^{-\frac{1}{4}}\right)} + \frac{9\sqrt{2}M\gamma}{N} \lambda^{-\frac{1}{4}},
\end{alignat}
where (a) comes from applying Lemmas~\ref{lem:ker_dist_weight} and \ref{lem:kernel_upperbound}. Therefore, we have:
\begin{align}\label{eq:ladv_final}
    \ldec^a(\hat{f}) &= 4\mathop{\sup}_{\mathcal{A}_\gamma}   \mathop{\sup}_{z}  \left|\decs\left(z\right) - \decs^\circ(z)\right|^2  \nonumber \\
    &\lec{}{\leqslant} 4  \left|\mu_3 \frac{2M\gamma}{N^2} + \mu_4\frac{2M\gamma}{N}\lambda^{-\frac{1}{4}}e^{\left(-\frac{1}{\sqrt{2}}\lambda^{-\frac{1}{4}}\right)} + \frac{9\sqrt{2}M\gamma}{N} \lambda^{-\frac{1}{4}}\right|^2  \nonumber \\
    &\lec{(a)}{\leqslant}  C_1 \frac{M^2\gamma^2}{N^4} + C_2\frac{M^2\gamma^2}{N^2}\lambda^{-\frac{1}{2}}e^{\left({-\sqrt{2}}\lambda^{-\frac{1}{4}}\right)} + C_3\frac{M^2\gamma^2}{N^2} \lambda^{-\frac{1}{2}},
\end{align}
where (a) follows by AM-GM inequality and  $C_1:=48(\mu_3)^2, C_2:=48(\mu_4)^2, C_3:=4\times 3 \times (9\sqrt{2})^2=1944$.

Combining \eqref{eq:ladv_final}, \eqref{eq:ldec_final}, and \eqref{eq:lenc_final} along with defining $C_4:=c_3^{\frac{3}{4}}$ and $C_5:=c_4^{\frac{3}{4}}$ complete the proof.

\section{Proof of Theorem~\ref{th:enc_bound}}\label{sec:app_th_enc}
According to Theorem~\ref{th:dec_bound}, we have:
\begin{alignat}{2}\label{eq:th2_bound_1}
    \mathcal{R}(\hat{f}) &\leqslant 
    &&C_1 \frac{M^2\gamma^2}{N^4} + C_2\frac{M^2\gamma^2}{N^2}\declamb^{-\frac{1}{2}}\left( e^{\left({-\sqrt{2}}\declamb^{-\frac{1}{4}}\right)} + C_3\right) + \left(C_4\declamb^\frac{3}{4} +C_5N^{-3}\right) \norm{(f \circ \encs)''}^{2}_{\lp{2}{\Omega}} + 
     \frac{2\nu^2}{K} \sum^K_{k=1} (\encs(\alpha_k) - x_k)^2
    \nonumber \\
    &\lec{(a)}{\leqslant}
    &&C_1 \frac{M^2\gamma^2}{N^4} + C_2(1+C_3)\frac{M^2\gamma^2}{N^2}\declamb^{-\frac{1}{2}} + \left(C_4\declamb^\frac{3}{4} +C_5N^{-3}\right) \norm{(f \circ \encs)''}^{2}_{\lp{2}{\Omega}} 
    + \frac{2\nu^2}{K} \sum^K_{k=1} (\encs(\alpha_k) - x_k)^2, 
\end{alignat}
where (a) holds because \( \declamb > 0 \). 

The following lemma provides a uniform upper bound for the complex term \( \norm{(f \circ \encs)''}^{2}_{\lp{2}{\Omega}} \) in terms of \( \norm{\encs}^2_{\hiltilde{2}} \).
\begin{lemma}\label{lem:foue_bound}
    If $f$ is $\nu$-Lipschitz function and $|f''|\leqslant \eta$, we have:
    \begin{align}
        \norm{(f \circ \encs)''}^{2}_{\lp{2}{\Omega}} \leqslant (\eta^2 + \nu^2)\cdot \xi\left(\norm{\encs}^2_{\hiltilde{2}}\right).
    \end{align}
\end{lemma}
\begin{proof}
    By applying the chain rule and following similar steps as those outlined in \cite[Theorem~3]{moradicoded}, we can demonstrate that:
\begin{align}\label{eq:optimal_encoder_1}
    \int_\Omega\left[f(\encs(t))''\right]^2 \,dt &\lec{(a)}{=} \int_\Omega \left[ \encs''(t)\cdot f'(\encs(t)) + \encs'(t)^2 \cdot f''(\encs(t))\right]^2\,dt \nonumber \\ 
    &\lec{(b)}{\leqslant} \int_\Omega \left[\encs''(t)^2 + \encs'(t)^4\right] \left[f'(\encs(t))^2 + f''(\encs(t))^2\right]\,dt \nonumber \\
    &\lec{(c)}{\leqslant} (\eta^2 +\nu^2) \int_\Omega \left[\encs''(t)^2 + \encs'(t)^4\right]\,dt \nonumber\\ 
    &\lec{}{=} (\eta^2 + \nu^2) \left(\norm{\encs''(t)}^2_{\lp{2}{\Omega}} + \norm{\encs'(t)}^4_{\lp{4}{\Omega}} \right) \nonumber\\
    &\lec{(d)}{\leqslant} (\eta^2 + \nu^2) \left(\norm{\encs''(t)}^2_{\lp{2}{\Omega}} + 
    \left(\norm{\encs'(t)}_{\lp{2}{\Omega}} + \norm{\encs''(t)}_{\lp{2}{\Omega}}\right)^4\right) \nonumber \\ 
    &\lec{(e)}{\leqslant} (\eta^2 + \nu^2) \left(\norm{\encs''(t)}^2_{\lp{2}{\Omega}} + 
    4\left(\norm{\encs'(t)}^2_{\lp{2}{\Omega}} + \norm{\encs''(t)}^2_{\lp{2}{\Omega}}\right)^2\right)
    \nonumber \\ 
    &\lec{(f)}{\leqslant} (\eta^2 + \nu^2) \left(\norm{\encs}^2_{\sob{2}{2}} + 
    4\norm{\encs}^4_{\sob{2}{2}}\right)
    \nonumber \\ 
    &\lec{(g)}{\leqslant} (\eta^2 + \nu^2) \left(7\norm{\encs}^2_{\sobeq{2}{2}} + 
    4\times7^2\norm{\encs}^4_{\sobeq{2}{2}}\right)
    \nonumber \\ 
    &\lec{(h)}{=} (\eta^2 + \nu^2) \cdot \xi\left(\norm{\encs}^2_{\sobeq{2}{2}}\right),
    \nonumber \\
    &\lec{(i)}{=} (\eta^2 + \nu^2) \cdot \xi \left(\norm{\encs}^2_{\hiltilde{2}}\right),
\end{align}
where (a) follows from the chain rule, (b) is derived using the Cauchy-Schwartz inequality, (c) relies on the bound \( \norm{f''}_{\lp{\infty}{\Omega}} \leqslant \eta \), (d) utilizes Theorem~\ref{th:interp_ineq} with \( r=4 \), \( p=q=2 \), and \( l=1 \), (e) is based on the AM-GM inequality, (f) results from adding positive terms \( \norm{\encs}^2_{\sob{2}{2}} \) and \( \norm{\encs'}^2_{\sob{2}{2}} \) to the first term and \( \norm{\encs}^2_{\sob{2}{2}} \) to the second term within the parentheses, (g) follows from applying Corollary~\ref{cor:norm-equiv-1}, (h) is derived by defining \( \xi(t):=7t + 4\times7^2t^2 \), and (i) is justified by Proposition~\ref{prop:sob_hilb}.
\end{proof}
As a result, by using Lemma~\ref{lem:foue_bound} and \eqref{eq:th2_bound_1}, we have:
\begin{align}
    \mathcal{R}(\hat{f}) &\lec{(a)}{\leqslant} \frac{2\nu^2}{K}\sum^K_{k=1} (\encs(\alpha_k) - x_k)^2 + C_1 \frac{M^2\gamma^2}{N^4} + \widetilde{C_2}\frac{M^2\gamma^2}{N^2}\declamb^{-\frac{1}{2}} + (\eta^2 + \nu^2)\left(\left(C_4\declamb\right)^\frac{3}{4} +\left( C_5N^{-4}\right)^\frac{3}{4}\right) \cdot \xi \left(\norm{\encs}^2_{\hiltilde{2}}\right) \nonumber \\ 
    &\lec{(b)}{\leqslant} \frac{2\nu^2}{K}\sum^K_{k=1} (\encs(\alpha_k) - x_k)^2 + \enclamb \cdot \left(2 +\xi \left(\norm{\encs}^2_{\hiltilde{2}}\right)\right)
    \nonumber \\ 
    &\lec{(c)}{\leqslant} \frac{2\nu^2}{K}\sum^K_{k=1} (\encs(\alpha_k) - x_k)^2 + \enclamb \cdot \psi \left(\norm{\encs}^2_{\hiltilde{2}}\right),
\end{align}
where (a) follows from $\widetilde{C_2}:=C_2(1+C_3)$, (b) follows by $\enclamb:=\max\{C_1 \frac{M^2\gamma^2}{N^4},
\widetilde{C}_2\frac{M^2\gamma^2}{N^2}\declamb^{-\frac{1}{2}},(\nu^2+\eta^2)(C_4\declamb^\frac{3}{4} + C_5N^{-3})\}$, (c) is comes from $\psi(t):=\xi(t) + 2$.

\section{Proof of Corollary~\ref{cor:conv_rate}}\label{sec:app_cor_rate_proof}

Consider a natural spline function, $\widetilde{\encs}(\cdot)$, fitted to the data points $\{\alpha_k, x_k\}_{k=1}^K$. Let $\encs^*(\cdot)$ represent the minimizer of the upper bound specified in \eqref{eq:optimcal_encdoer}. This leads to the following inequality:

\begin{align}\label{eq:proof_cor_1}
    \mathcal{R}(\hat{f}) &\leqslant \frac{2\nu^2}{K} \sum_{k=1}^K \left(\encs^*(\alpha_k) - x_k\right)^2 + \enclamb \cdot \psi \left(\norm{\encs^*}^2_{\hiltilde{2}}\right)
    \nonumber \\
    &\lec{(a)}{\leqslant} \frac{2\nu^2}{K} \sum_{k=1}^K \left(\widetilde{\encs}(\alpha_k) - x_k\right)^2 + \enclamb \cdot \psi \left(\norm{\widetilde{\encs}}^2_{\hiltilde{2}}\right) \nonumber \\
    &\stackrel{(b)}{\leq} \max\left\{C_1 \frac{M^2\gamma^2}{N^4},
\widetilde{C}_2\frac{M^2\gamma^2}{N^2}\declamb^{-\frac{1}{2}},(\nu^2+\eta^2)(C_4\declamb^\frac{3}{4} + C_5N^{-3})\right\} \cdot \psi \left(\norm{\encs^*}^2_{\hiltilde{2}}\right),
\end{align}
where (a) follows from the optimality of $\encs^*(\cdot)$, and (b) results from substituting the definition of $\enclamb$ and the fact that the natural spline overfits the data, i.e., $\widetilde{\encs}(\alpha_k) = x_k$ for all $k \in [K]$.

The upper bound in \eqref{eq:proof_cor_1} holds for all \( 1 \geqslant \declamb > CN^{-4} \). By optimizing over $\lambda$, we conclude that 
\(\declamb^* = J N^{\frac{8}{5}(a-1)}\), where the constant $J$ depends on $\widetilde{C}_2, C_4, C_5, \nu, \eta$, and $M$. Substituting $\declamb^*$ yields:
\begin{align}
    \mathcal{R}(\hat{f}) \leqslant \max\left\{C_1 M^2 N^{2a-4},\,
    \widetilde{C}_2 M^2 J^{-\frac{1}{2}} N^{\frac{6}{5}(a-1)},\,
    (\nu^2 + \eta^2)\left(C_4 J^{\frac{3}{4}} N^{\frac{6}{5}(a-1)} + C_5 N^{-3}\right)\right\}.
\end{align}
Since \( a \in [0,1) \), we have \( \frac{6}{5}(a-1) \geqslant \max\{2a-4, -3\} \). Consequently, 
\(\mathcal{R}(\hat{f}) \leqslant \mathcal{O}\left(N^{\frac{6}{5}(a-1)}\right)\).

\section{Proof of Theorem~\ref{th:letcc_enc_ss}}\label{sec:app_th_enc_ss}
For the first part, note that if $t\leqslant E$, we have:
\begin{align}
    \psi(t) &\lec{}{=} 2 + \xi(t) \nonumber \\
    &\lec{}{\leqslant} 2 + t(7 + 196E) \nonumber \\
    &\lec{}{=} c_1(E) + c_2(E)t.
\end{align}
Therefore, we have:
\begin{align}
    \psi\left(\norm{\encs}^2_{\hiltilde{2}}\right) &\lec{}{\leqslant} c_1(E) + c_2(E)\left(\norm{\encs}^2_{\hiltilde{2}}\right) \nonumber \\
    &\lec{(a)}{=} c_1(E) + c_2(E)\left(\encs(0)^2 + \encs'(0)^2 + \int_{\Omega} |\encs''(t)|^2\,dt\right) \nonumber \\
    &\lec{(b)}{\leqslant} c_1(E) + c_2(E)\left(E + \int_{\Omega} |\encs''(t)|^2\,dt\right),
\end{align}
where (a) follows from the definition of $\norm{\cdot}^2_{\hiltilde{2}}$, and (b) holds because $\norm{\encs}^2_{\hiltilde{2}} \leqslant E$. Defining $D_1(E) := c_1(E) + E \cdot c_2(E)$ and $D_2(E) := c_2(E)$ completes the proof of the first part.

For the second part, the same steps as those in \cite[Proposition~1]{moradicoded} can be followed to prove the statement.
\end{document}